\documentclass[%
superscriptaddress,
amsmath,amssymb,amsfonts,
twocolumn,
prx,
aps,
]{revtex4-2}

\pdfoutput=1 
\usepackage{graphicx}
\usepackage{bm} 
\usepackage{physics}
\usepackage{amsthm}
\usepackage{amsmath}
\usepackage{xcolor}
\usepackage{textcomp}
\usepackage{outlines}
\usepackage{soul}
\usepackage[colorlinks=true, allcolors=black]{hyperref}
\usepackage[toc,page]{appendix}
\usepackage{cleveref}

\newcommand{\G}{\mathbb{G}}

\newcommand{\Lag}{\mathcal{L}}
\newcommand{\I}{\mathbb{I}}
\newcommand{\Id}{\mathbb{I}}

\newcommand{\T}{\mathbb{T}}
\newcommand{\U}{\mathbb{U}}

\newcommand{\EX}{\mathbb{E}}

\newcommand{\Asym}{\text{Asym}}
\newcommand{\ve}{\vb{e}}
\newcommand{\mA}{\vb{A}}
\newcommand{\mB}{\vb{B}}
\newcommand{\mI}{\vb{I}}
\newcommand{\cI}{\mathcal{I}}
\newcommand{\mH}{\vb{H}}
\newcommand{\cH}{\mathcal{H}}
\newcommand{\vj}{\vb{j}}
\newcommand{\mJ}{\vb{J}}
\newcommand{\J}{\mathbb{J}}
\newcommand{\tj}{\tilde{\vj}}
\newcommand{\vn}{\vb{n}}
\newcommand{\vx}{\vb{x}}
\newcommand{\vy}{\vb{y}}
\newcommand\numthis{\stepcounter{equation}\tag{\theequation}}
\usepackage{mathtools}
\DeclarePairedDelimiter\ceil{\lceil}{\rceil}
\DeclarePairedDelimiter\floor{\lfloor}{\rfloor}

\newcommand{\bmm}[1]{\mathbb{#1}}

\newtheorem{theorem}{Theorem}[section]
\newtheorem{lemma}[theorem]{Lemma}

\newcommand{\REVISION}[1]{{\color{red} #1}}
\DeclareUnicodeCharacter{202F}{\REVISION{There is a 202F unicode character here}}

\begin{document}

\preprint{APS/123-QED}

\title{Maximum Shannon Capacity of Photonic Structures}

\author{Alessio~Amaolo\(^\ddagger\)}
\email{Contact author: alessioamaolo@princeton.edu}
\affiliation{Department of Chemistry, Princeton University, Princeton, New Jersey 08544, USA}
\thanks{These authors contributed equally to this work.}
\author{Pengning~Chao\(^\ddagger\)}
\email{Contact author: pchao827@mit.edu}
\affiliation{Department of Mathematics, Massachusetts Institute of Technology, Cambridge, Massachusetts 02139, USA}
\thanks{These authors contributed equally to this work.}
\author{Benjamin~Strekha}
\affiliation{Department of Electrical and Computer Engineering, Princeton University, Princeton, New Jersey 08544, USA}
\author{Stefan~Clarke}
\affiliation{Department of Operations Research and Financial Engineering, Princeton University, Princeton, New Jersey 08544, USA}
\author{Jewel~Mohajan}
\affiliation{Department of Electrical and Computer Engineering, Princeton University, Princeton, New Jersey 08544, USA}
\author{Sean~Molesky}
\affiliation{Department of Engineering Physics, Polytechnique Montréal, Montréal, Québec H3T 1J4, Canada}
\author{Alejandro~W.~Rodriguez}
\affiliation{Department of Electrical and Computer Engineering, Princeton University, Princeton, New Jersey 08544, USA}
\date{\today}

\begin{abstract}
Information transfer through electromagnetic waves is an important problem that touches a variety of technologically relevant applications, including computing, telecommunications, and power management. 
Prior attempts to establish limits on optical information transfer have focused exclusively on waves propagating through vacuum or a known photonic structure (prescribed wave sources and receivers).
In this article, we describe a mathematical theory that addresses fundamental questions concerning optimal information transfer in photonic devices. 
Combining information theory, wave scattering, and optimization theory, we formulate bounds on the maximum Shannon capacity that may be achieved by structuring senders, receivers, and their environment.
This approach provides a means to understand how material selection, device size, and general geometrical features impact power allocation, communication channels, and bit-rate in photonics.  
Allowing for arbitrary structuring leads to a non-convex problem that is significantly more difficult than its fixed structure counterpart, which is convex and satisfies a known ``water-filling'' solution.
We derive a geometry-agnostic convex relaxation of the problem that elucidates fundamental physics and scaling behavior of Shannon capacity with respect to device parameters and the importance of structuring for enhancing capacity.
We also show that in regimes where communication is dominated by power insertion requirements, bounding Shannon capacity maps to a biconvex optimization problem in the basis of singular vectors of the Green's function. 
This problem admits analytic solutions that give physically intuitive interpretations of channel and power allocation and reveals how Shannon capacity varies with signal-to-noise ratio.
Proof of concept numerical examples show that bounds are within an order of magnitude of achievable device performance and successfully predict the scaling of performance with channel noise.
The presented methodologies have implications for the optimization of antennas, integrated photonic devices, metasurface kernels, MIMO space-division multiplexers, and waveguides for maximizing communication efficiency and bit-rates.
\end{abstract}

\maketitle

\section{Introduction}

From the 19th century electrical telegraph to modern 5G telecommunications, electromagnetics has played an increasingly important role in human communication. 
Paralleling this rapid improvement in physical technologies, our understanding of the fundamental meaning of information has also steadily improved.
Building upon earlier work by Nyquist~\cite{Nyquist_1924} and Hartley~\cite{Hartley_1928}, Shannon demonstrated in two landmark papers~\cite{shannon_mathematical_1948,shannon_communication_1949} that not only is it possible to transmit information at a finite rate with arbitrarily small error through noisy channels, but that the maximum transmission rate---now known as the Shannon capacity---can also be explicitly calculated as a function of the channel bandwidth and signal-to-noise ratio (SNR).  
Shannon's original derivation focused on scalar, time-varying signals subject to additive white Gaussian noise, but the principles he elucidated are universal and can be applied to more general physical settings, including electromagnetic wave phenomena~\cite{franceschetti_wave_2017}. 
Information capacity results have been obtained for a variety of systems including fiber optics~\cite{tang_shannon_2001,essiambre_capacity_2010,shtaif_challenges_2022} and wireless communications networks~\cite{franceschetti_capacity_2009,lee_capacity_2012,goldsmith_shannon_2011}; in particular, an increasingly relevant area of research are multiple-input multiple-output (MIMO) systems that make use of spatial multiplexing, wherein the signal fields have spatial as well as temporal degrees of freedom~\cite{gesbert_theory_2003}. 
Communication in the presence of a scatterer has also been analyzed by choosing an appropriate optical characteristic proxy, e.g., the number of waves escaping a region~\cite{miller_kuang_miller_tunneling_2023} or the sum of field amplitudes at the receiver~\cite{Miller_2013}.
These figures of merit are intimately related to other quantities of power transfer such as thermal radiation~\cite{molesky_mathbbt-operator_2019} and near-field radiative heat transfer~\cite{molesky_fundamental_2020,venkataram_fundamental_2020-1} and establish intuitive connections between fields at the receiver and communication, but do not directly maximize Shannon capacity. 

These prior results generally work well if the wave propagation medium is fixed, be it antenna design in free-space \cite{ehrenborg_bounds_MIMO_2017,ehrenborg_bounds_MIMO_2020,ehrenborg_capacity_bounds_2021} or guided modes in a fiber~\cite{tang_shannon_2001,essiambre_capacity_2010,shtaif_challenges_2022}. 
Yet, they do not capture the range of physics and possibilities of control and optical response achievable via (nano)photonic structuring. 
Antennas at either the sender or receiver may be designed to enhance gain and directivity (therefore bit-rate), and similarly their environment may be structured to achieve greater communication and field enhancements (e.g., parabolic reflectors, metasurfaces, waveguides, fibers, etc.).
In fact, there is emerging consensus on the increasing importance of complex structured environments for electromagnetic communication in a wide range of contexts, including silicon photonics~\cite{shekhar_roadmapping_2024}, 
on-chip optical interconnects~\cite{caulfield_supercomputing_2010,bashir_survey_2019}, and integrated image processing~\cite{sitzmann_endtoend_2018,lin_endtoend_2021,lin_endtoend_2022,hazineh_DFlat_2022}. Naturally, for each of these systems, the question of \textit{how much} material structuring can improve their information transfer capabilities is key to quantifying future potential. 
Increasing application of computational photonic design~\cite{molesky_inverse_2018} utilizing techniques from mathematical optimization to maximize field objectives over many structural degrees of freedom has begun to address performance gaps in many settings. In particular, such design freedom and increased ability to concentrate power into selective wavelengths (e.g., resonant modes) or to manipulate fields in counter-intuitive ways has a wide array of applications and implications for photonic templates of communication: metasurface kernels~\cite{molesky_t_2022}, integrated resonators~\cite{ahn_microresonators_2022}, and channel engineering~\cite{Miller_2013}. 

In this work, we leverage recent developments in photonic optimization to present a unified framework for investigating the extent to which structured photonic devices may enhance information transfer. First, we show how the structure-dependent, non-convex problem of maximizing information transfer can be relaxed into a shape-independent, convex problem to arrive at a bound on Shannon capacity subject to drive-power and physical wave constraints derived from Maxwell's equations. 
Second, we show that in regimes where power requirements are dominated by insertion impedance (i.e., internal resistance or contact impedance in the driving currents), the maximization of the Shannon capacity can be significantly simplified in the basis of singular vectors of the Green's function, written in terms of a single signal-to-noise free parameter, and further relaxed to a biconvex problem that admits analytic solutions.  
We present several illustrative examples where we compare these bounds to inverse designs and show that not only do the bounds predict trends but are within an order of magnitude of achievable performance for a wide range of signal-to-noise ratios. 

Bounds on the Shannon capacity provide performance targets for communication and elucidate how information capacity scales with key device parameters (e.g., material choices, device sizes, or bandwidth). 
The fundamental limitations encoded in Maxwell's equations (spectral sum rules, finite speed of light, space-bandwidth products, etc.) dictate a finite limit on the ability to concentrate power over selective wavelengths (e.g., via resonances) or propagate fields over a given distance~\cite{chao_physical_2022, barnett_sum_1996, scheel_sum_2008, shim_fundamental_2019,zhang_all_2023}. The proposed method paves the way for understanding how such trade-offs affect electromagnetic information transfer.

\section{Theory} \label{sec:theory}

\begin{figure}
    \centering
    \includegraphics[width=0.9\linewidth]{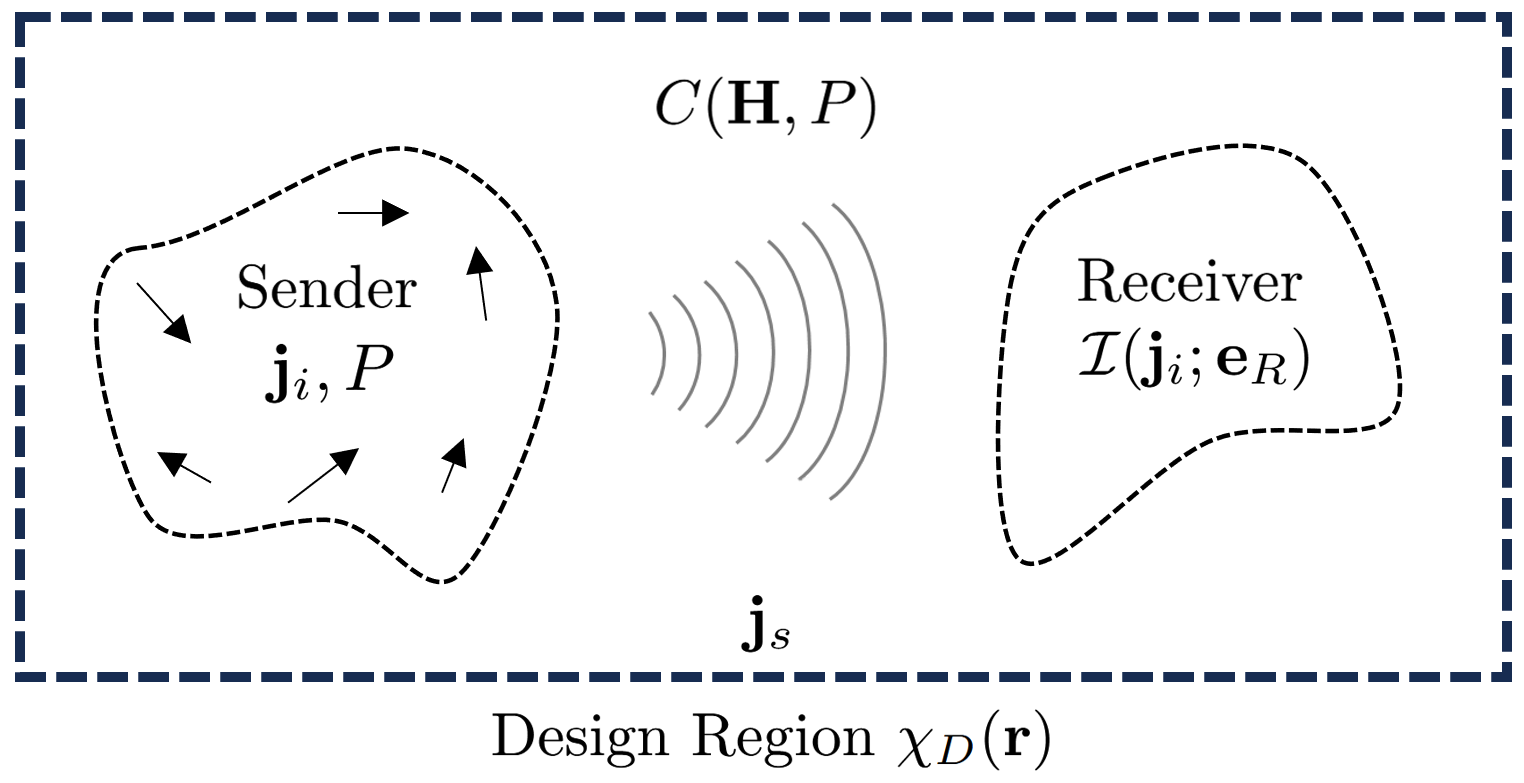}
    \caption{Schematic of photonic communication setup. Sender and receiver devices (e.g., ``antennas", waveguides, free space) and their surrounding medium (e.g., metasurfaces, on-chip multiplexers, free space) may be simultaneously designed to maximize information transfer. 
    The Shannon capacity $C$ is the maximum achievable rate of error-free information transfer between input currents and receiving fields. 
    Free currents in the sender \(\vj_i\) and bound currents in the design \(\vj_s\) are optimized to achieve maximum $C$ respecting Maxwell's equations under a prescribed (limited) power budget.}
    \label{fig:schematic}
\end{figure}

In this section, we begin with a brief review of several defining relations of information capacity that pertain to wave communication. We show how the standard formula for MIMO communication can be derived from the fundamental definition of the Shannon capacity, subject to a constraint on the magnitude of the input signal. 
This constraint manifests electromagnetically as a constraint on input current amplitudes; we show how we can also derive a constraint with respect to drive power, connecting to prior work on vacuum information transfer in the context of antenna design~\cite{ehrenborg_bounds_MIMO_2017,ehrenborg_bounds_MIMO_2020,ehrenborg_capacity_bounds_2021}. 
By incorporating the freedom to structure senders, receivers, and their surrounding environments, we formulate a non-convex structural optimization problem that bounds the maximum Shannon capacity over any photonic structure. 
We show how the problem can be relaxed to a convex optimization problem or, alternatively, to a biconvex problem on the channel capacities as given by the electromagnetic Green's function.

We consider a MIMO system with input signal $\vb{x} \in \mathbb{C}^n$ and output signal $\vb{y} \in \mathbb{C}^m$ related by 
\begin{equation}
    \vy = \mH \vx + \vn
    \label{eq:MIMOchannel}
\end{equation}
where the $\mH \in \mathbb{C}^{m \times n}$ is the channel matrix and the channel has additive zero-mean complex Gaussian noise with covariance $\EX[\vn \vn^\dagger] = N\mI$, where $\EX[]$ denotes expected value. 
The channel matrix and noise prescription relate ``input" signals \(\vb x\) to ``output" signals \(\vb y\), which are viewed as random variables given the randomness of noise and the probabilistic argument used by Shannon to derive the Shannon capacity~\cite{shannon_mathematical_1948,shannon_communication_1949,mackay_inference_2019}. Communication is possible when $\vb x$ and $\vb y$ are \textit{dependent}: an observation of $\vb y$ changes the \textit{a posteriori} probability for any given $\vb x$. Denote the joint probability distribution of $\vb x$ and $\vb y$ as $\gamma(\vb x, \vb y)$; the marginal and conditional probability distributions are given by 
\begin{subequations}
    \begin{align}
        \gamma_{\vb x}(\vb x) &= \int \gamma(\vb x, \vb y) \dd \vb y, \\
        \gamma(\vb x|\vb y) &= \gamma(\vb x, \vb y) / \gamma_{\vb y}(\vb y), 
    \end{align}
\end{subequations}
and symmetrically in \(\vy\).
The differential and conditional differential entropy of \(\vx\) can now be written as
\begin{equation} \label{eq:entropy} \begin{aligned}
    \cH(\vx) &= \int \gamma_{\vb x}(\vx) \log_2(1/\gamma_{\vx}(\vx)) \dd \vx,~\text{and} \\
    \cH(\vx|\vy) &= \int \gamma_{\vb y}(\vy) \int \gamma(\vx|\vy) \log_2(1/\gamma(\vx|\vy)) \dd \vx \dd \vy.
\end{aligned} \end{equation}
$\cH(\vx)$ is a continuous analog of the discrete information entropy which encodes the level of uncertainty in \(\vb x\) 
and is a measure of the information gained upon observation of $\vb x$; $\cH(\vx|\vy)$ in turn is the expected entropy of \(\vx\) conditioned on observation of \(\vy\)~\cite{mackay_inference_2019}. 
Their difference is the \textit{mutual information} between \(\vx\) and \(\vy\), denoted \(\cI(\vx ; \vy)\): 
\begin{equation}
    \cI(\vx;\vy) = \cH(\vx) - \cH(\vx|\vy) = \cH(\vy) - \cH(\vy|\vx).
    \label{eq:mutual_info}
\end{equation}
Intuitively, $\cI(\vx;\vy)$ is the expected reduction in uncertainty (and hence gain in information) of $\vb x$ upon observation of $\vy$. The second equality follows from Eq.~\eqref{eq:entropy} and indicates that $\cI(\vx;\vy)$ is symmetric between $\vx$ and $\vy$; note that for our MIMO channel (\ref{eq:MIMOchannel}) the second term is simply \(\cH(\vy | \vx) = \cH(\vn)\).
Shannon's noisy coding theorem states that a tight upper limit on the error-free information transfer rate across the channel achievable via encoding is the Shannon capacity $C = \sup_{\gamma_{\vx}} \, \cI(\vx;\vy)$~\cite{shannon_mathematical_1948,cover_information_2006,mackay_inference_2019}. 
Over a channel with continuous signals and finite additive noise as in~\eqref{eq:MIMOchannel}, $C$ diverges without further constraints on the magnitude of \(\vx\): 
intuitively, a larger signal magnitude will make the corruption of additive noise $\vb{n}$ proportionally smaller. Without further details of the underlying physics of the channel, this is addressed by enforcing a cap on the average magnitude of $\vx$, which is equivalent to a constraint on the covariance $\vb{Q}$ of $\gamma_{\vx}$:
\begin{equation}
    \EX_{\gamma_{\vx}}[\vx^\dagger \vx] = \Tr(\EX_{\gamma_{\vx}}[\vx \vx^\dagger]) \equiv \Tr \vb Q \leq P.
    \label{eq:tr_constraint}
\end{equation}

The Shannon capacity of our MIMO system is thus
\begin{equation}
    C(\mH, P) = \sup_{\gamma_{\vx}, \Tr(\vb Q)\leq P} \quad \cI(\vx;\vy)
    \label{eq:coding_thm}
\end{equation}
where the supremum is taken over all possible probability distributions $\gamma_{\vx}(\vx)$ for $\vx$.
Given (\ref{eq:MIMOchannel}), one can show that among all $\gamma_{\vx}$ with a particular covariance $\vb{Q}$, $\cI(\vx;\vy)$ is maximized by the corresponding Gaussian distribution with covariance \(\vb Q \)~\cite{telatar_capacity_Gaussian_channels_1999}. 
Thus when solving (\ref{eq:coding_thm}) we can consider only Gaussian distributions, leading to the well-known expression for the Shannon capacity per unit bandwidth, \(C\), for a MIMO  system~\cite{gesbert_theory_2003, goldsmith_capacity_MIMO_2003}
\begin{subequations} \begin{align}
    C(\mH,P) &= \max_{\vb{Q} \succ \vb{0}, \Tr(\vb{Q})\leq P} \quad \cH(\vy) - \cH(\vn) \\
    &= \max_{\vb{Q} \succ \vb{0}, \Tr(\vb{Q})\leq P} \quad\log_2 \det(\mH\vb{Q}\mH^\dagger + N\mI) \nonumber \\ &\hspace{1.2in}- \log_2 \det(N\mI) \\
    &= \max_{\vb{Q} \succ \vb{0}, \Tr(\vb{Q})\leq P} \quad \log_2 \det(\mI + \frac{1}{N} \mH\vb{Q}\mH^\dagger). 
    \label{eq:capacity_MIMO}
\end{align} \end{subequations}
This is a convex optimization problem and has an analytic solution that is often referred to as ``water-filling."
For further mathematical details the reader is referred to~\cite{telatar_capacity_Gaussian_channels_1999} and Appendix~\ref{asec:waterfilling}. 

We will now consider Shannon capacity in the physical context of information transfer via electromagnetic fields and describe how to incorporate the possibility of photonic structuring into bounds on the information transfer rate. 
The schematic setup is illustrated in Fig.~\ref{fig:schematic}, where the input $\vx$ is a free current distribution $\vj_i$ within a sender region $S$, and the output $\vy$ is the electric field $\ve_R$ generated in a receiver region $R$. 
Working in the frequency domain with dimensionless units of $\mu_{0}=\epsilon_{0}=1$ and vacuum impedance \(Z = \sqrt{\mu_0 / \epsilon_0} = 1\), the electric field \(\vb e\) in a region obeys Maxwell's wave equation \(\left( \nabla \times \nabla \times - \epsilon(\vb r) \omega^2 \right) \ve(\vb r; \omega) = i\omega \vj_i(\vb r)\), where \(\omega\) is the angular frequency of the source current \(\vj_i\), \(k= \omega\), and \(\epsilon(\vb r)\) is the permittivity distribution in space. 
The total Green's function \(\G_t(\vb{r},\vb{r'}, \epsilon(\cdot); \omega)\) maps input currents \(\vb j_i\) to electric fields $\ve$ via $\ve(\vb r) = \frac{i}{\omega} \int \G_t(\vb r,\vb r') \vj_i(\vb r') \dd \vb r'$ and satisfies \( \left( \nabla \times \nabla \times - \epsilon(\vb r) \omega^2 \right)\G_t(\vb r, \vb r',\epsilon(\cdot);\omega) = \omega^2 \vb{I} \delta(\vb{r} - \vb{r'}) \), where $\vb I$ is the unit dyad. 
We also denote \(\chi_D\) a constant susceptibility factor associated with a photonic device and related to the permittivity profile via \(\epsilon(\vb r) - \I = P_\epsilon(\vb r) \chi_D\) where \(P_\epsilon(\vb r)\) is a projection operator containing the spatial dependence of \(\epsilon(\vb r)\).
In vacuum, there is no photonic structuring so the vacuum Green's function is given by \(\G_0 \equiv \G_t(\vb r, \vb r', \Id; \omega) \). We denote \(\G_{t,RS}\) as the sub-block of the Green's function that maps currents in the source region to fields in the receiver region, so the channel matrix \(\vb H = \frac{i}{\omega} \G_{t,RS}\)~\cite{zhu_electromagnetic_info_theory_2024}. 
As shown in Fig.~\ref{fig:schematic}, we can modify $\G_t$ and $\vb{H}$ through the engineering of a photonic structure $\epsilon(\vb r)$ within the design region $D$, which may also incorporate either the $S$ or $R$ regions; the central question we seek to address is the extent to which photonic engineering may boost the Shannon capacity $C$. 

Some constraints on the input currents $\vj_i$ are needed for $C$ to be finite; instead of the magnitude constraint prescribed in the more abstract formulation of~\eqref{eq:coding_thm}, we instead consider a budget $P$ on the total drive power needed to maintain $\vj_i$: 
\begin{equation} \label{eq:drive_power_constraint}
    \EX\left[-\frac{1}{2}\Re{\frac{i}{\omega} \vj_i^\dagger \G_{t,SS} \vj_i} \right] + \EX\left[ \alpha |\vj_i|^2 \right] \leq P.
\end{equation} 
The total drive power in~\eqref{eq:drive_power_constraint} is comprised of two parts: the first term is the power extracted from $\vj_i$ by the electric field, which includes radiated power and material absorption: optical loss processes that are modeled explicitly via Maxwell's equations. 
The second term is an ``insertion cost'' for the (abstract) drive producing the input currents, which we assume to be proportional to the squared norm of $\vj_i$ with proportionality factor $\alpha$.
Depending on the specific drive mechanism, this term may represent, for example, the internal resistance of a power source or contact impedance leading into an antenna~\cite{ehrenborg_capacity_bounds_2021}. 

Adapting the Shannon capacity expression \eqref{eq:coding_thm} to the specifics of our electromagnetic channel along with the drive power budget \eqref{eq:drive_power_constraint}, keeping in mind that we are also designing the propagation environment (including sender and receiver) via $\epsilon(\vb r)$, we arrive at the following \textit{structural} optimization problem
\begin{subequations}
\begin{align}
    \max_{\epsilon(\vb{r})} \quad &C(\epsilon, P) = \max_{\gamma_{\vj_i}} \, \cI(\vj_i;\ve_R;\epsilon) \\
    \text{s.t.} \quad & \EX_{\gamma_{\vj_i}} \left[-\frac{1}{2}\Re{\frac{i}{\omega} \vj_i^\dagger \G_{t,SS}(\epsilon) \vj_i}  + \alpha|\vj_i|^2\right] \leq P.
\end{align}
\label{eq:Shannon_TO}
\end{subequations}
Given that \eqref{eq:drive_power_constraint} can be written as a linear trace constraint on the covariance of $\gamma_{\vj_i}$, we only need to consider Gaussian distributions for the inner optimization, allowing us to rewrite \eqref{eq:Shannon_TO} more explicitly as 
\begin{subequations} \label{eq:Shannon_GaussianTO}
    \begin{align}
        \max_{\epsilon(\vb{r}), \J_i} \quad & \log_2 \det(\Id + \dfrac{P}{N\omega^2} \G_{t, RS}(\epsilon) \J_i \G_{t, RS}^\dagger (\epsilon)) \label{eq:shannon_objective} \\
        \text{s.t.} \quad & \Tr(\left( \alpha \I_{S} + \dfrac{1}{2\omega} \Asym\G_{t, SS}(\epsilon)\right) \J_i) \leq 1 \label{eq:shannon_objective_power} \\
        & \J_i \succeq \vb{0} \\
        & (\curl\curl - \epsilon(\vb{r})\omega^2) \G_t(\vb{r},\vb{r}';\epsilon) = \omega^2 \mI \delta(\vb{r}-\vb{r}')
    \end{align}
\end{subequations}
where \(\J_i\) is the covariance operator \(\EX[\vj_i \vj_i^\dagger]/P\) and \(\Asym \vb{A} \equiv \frac{1}{2i} (\vb{A} - \vb{A}^\dagger)\).
This scaling by \(P\) allows the expression of the objective in terms of a familiar signal-to-noise ratio \(P/N\). The two free parameters in Eq.~\eqref{eq:Shannon_GaussianTO} are \(P/N\) and \(\alpha\). 
We note that \(\alpha\) formally has units of resistivity: in this article, we take these units to be \(Z \lambda\). Similarly, \(P/N\) has units of \(\epsilon_0 c \lambda^2\) in 3D and \(\epsilon_0 c \lambda\) in 2D.
Note that the optimization degrees of freedom \(\epsilon(\vb r)\) may be restricted to any sub-region of the domain (e.g., one may consider a prescribed set of antennas and instead seek only to optimize the region between them).
This problem is non-convex owing to the nonlinear dependence of \(\G\) on \(\epsilon(\vb r)\). 
In particular, the water-filling solution does not apply: maximizing Shannon capacity involves co-optimization of both current covariance \textit{and} the photonic structure. 
Lastly, we note that while Eq.~\eqref{eq:Shannon_GaussianTO} contains functions of continuous operators, its corresponding discretized version converges with increasing resolution (see Appendix~\ref{asec:discretization}).

In the following sections, we will derive bounds on \eqref{eq:Shannon_GaussianTO} via two different methods and discuss our results in the context of prior work. In the first method of Sec.~\ref{sec:convex_relax}, we will reinterpret~\eqref{eq:Shannon_TO} from a maximization over $\epsilon(\vb r)$ and $\gamma_{\vj_i}$ to a maximization over joint probability distributions of $\vj_i$ and the induced polarization current $\vj_s = \left( \epsilon(\vb r) - \I \right) \G_t(\epsilon(\vb r)) \vj_i$. 
By relaxing physical consistency requirements on this joint probability distribution, we obtain a convex problem that can be solved with methods in semidefinite programming (SDP) or via a convex duality approach described in Appendix~\ref{asec:SDP_numerics}. 
Second, in Sec.~\ref{sec:biconvex_relax}, we consider a regime in which the drive power~\eqref{eq:drive_power_constraint} is dominated by insertion losses in the sender, and utilize a convenient basis transformation to express the Shannon capacity in terms of the singular values of \(\G_{t,RS}\), yielding an intuitive decomposition into orthogonal channels. 
By relaxing Maxwell's equations to encode wave behavior directly into structurally agnostic bounds on these singular values, we derive a \textit{biconvex} problem that for a limited number of constraints admits analytic solutions and physically intuitive interpretations of channel and power allocation. 
Both relaxations can in principle be applied in the appropriate regimes of validity, but theoretically understanding their differences is challenging: while both encode the physics of Maxwell's equations, the first does so via power-conservation constraints on polarization currents while the second through the dependence of the channel matrix on its singular values. 
Surprisingly, however, we find that the predicted bounds are close in regimes where they can be compared. 

\subsection{Convex Relaxation via Expected Power Conservation over Joint Distributions} 
\phantomsection \label{sec:convex_relax}

For notational convenience we define the combined source-induced current $\tj = \begin{pmatrix}\vb{j}_i \\ \vb{j}_s \end{pmatrix}$ and a function that maps input currents to induced polarization currents \(S(\vj_i; \epsilon) = (\epsilon-\Id) \G_t(\epsilon) \vj_i\). 
The set of physically feasible joint distributions is then given by
\begin{equation}
    \mathcal{P}_{physical} = \left\{\gamma_{\tj} \,\bigg|\, \exists \epsilon \,  \forall \vj_i \quad \gamma_{\tj}(\vj_s|\vj_i) = \delta(\vj_s - S(\vj_i; \epsilon)) \right\}.
    \label{eq:physical_set}
\end{equation}
Sampling any joint distribution $\gamma_{\tj} \in \mathcal{P}_{physical}$ gives $\tj$ such that $\vj_s$ is the induced polarization current in some $\epsilon(\vb{r})$ given $\vj_i$; furthermore for a fixed $\gamma_{\tj}$ this underlying $\epsilon(\vb{r})$ is consistent across all $\vj_i$.  
Given $\vj_s$, the output field $\ve_R$ and the extracted drive power $P$ can be expressed solely using the vacuum Green's function
\begin{gather}
    \ve_R = \frac{i}{\omega}\G_{RS}(\epsilon) \vj_i = \frac{i}{\omega}(\G_{0,RS}\vj_i + \G_{0,RD}\vj_s), \\
P = -\frac{1}{2}\Re{\frac{i}{\omega}\vj_i^\dagger (\G_{0,SS}\vj_i + \G_{0,SD} \vj_s)} + \alpha |\vj_i|^2.
\end{gather} 
Observe that the joint distribution optimization
\begin{subequations}     \label{eq:Shannon_joint} %
    \begin{align}
        &\max_{\gamma_{\tj} \in \mathcal{P}_{physical}} \, \cI(\tj;\ve_R)  \\
        &\text{s.t.} \, \EX_{\gamma_{\tj}}\left[-\frac{1}{2}\Re{\frac{i}{\omega}\vj_i^\dagger (\G_{0,SS}\vj_i + \G_{0,SD} \vj_s) } + \alpha |\vj_i|^2 \right] \leq P
    \end{align}
\end{subequations}
is exactly equivalent to \eqref{eq:Shannon_TO}: here, the $(\epsilon(\vb r), \gamma_{\vj_i})$ pairs have a one-to-one correspondence with $\gamma_{\tj} \in \mathcal{P}_{physical}$, both producing the same output distribution $\gamma_{\ve_R}$ and hence, the same $\mathcal{H}(\ve_R)$ and $\mathcal{I}$. 

So far, \eqref{eq:Shannon_joint} is just a rewriting of \eqref{eq:Shannon_TO} with the complexity of the problem buried in the description of the set $\mathcal{P}_{physical}$; we now relax~\eqref{eq:Shannon_joint} by considering a superset of $\mathcal{P}_{physical}$ that has a simpler mathematical structure. 
Instead of insisting that there exists some specific $\epsilon(\vb{r})$ for a member joint distribution as in (\ref{eq:physical_set}), we simply require that \textit{on average}, $\vj_i$ and $\vj_s$ satisfy structure agnostic local energy conservation constraints (see Appendix~\ref{asec:QCQP}):
\begin{equation} \label{eq:poynting}
    \EX_{\gamma_{\tj}}[\vj_i^\dagger \G_{0,DS}^\dagger \Id_k \vj_s - \vj_s^\dagger \U \Id_k \vj_s] = 0
\end{equation}
where \(\U \equiv \left( \chi_D^{-\dagger}\I - \G_{0,DD}^\dagger \right)\), $\Id_k$ is the indicator function over some design sub-region $D_k \subset D$ and \(\chi_D\) is the \textit{constant} susceptibility of the design region that does not depend on geometry. 
This relaxed requirement gives us the relaxed distribution set
\begin{equation}
    \mathcal{P}_{relaxed} = \left\{\gamma_{\tj} \,\bigg|\, \EX_{\gamma_{\tj}}[\vj_i^\dagger \G_{0,DS}^\dagger \Id_k \vj_s - \vj_s^\dagger \U \Id_k \vj_s] = 0\right\}
    \label{eq:relaxed_set}
\end{equation}
with an associated Shannon capacity problem
\begin{equation}
    \begin{aligned}
        &\max_{\gamma_{\tj} \in \mathcal{P}_{relaxed}} \, \cI(\tj;\ve_R) \\
        &\text{s.t.} \, \EX_{\gamma_{\tj}}\left[-\frac{1}{2}\Re{\frac{i}{\omega}\vj_i^\dagger (\G_{0,SS}\vj_i + \G_{0,SD} \vj_s)} + \alpha |\vj_i|^2\right] \leq P.
    \end{aligned}
    \label{eq:Shannon_relaxed}
\end{equation}
Given that $\mathcal{P}_{physical} \subset \mathcal{P}_{relaxed}$, the optimum of \eqref{eq:Shannon_relaxed} bounds \eqref{eq:Shannon_joint} and hence \eqref{eq:Shannon_TO} from above. Furthermore, the restrictions of \eqref{eq:relaxed_set} are just linear constraints on the covariance of $\gamma_{\tj}$; we can again restrict attention to just Gaussian distributions to solve \eqref{eq:Shannon_relaxed}, yielding
\begin{subequations}
    \begin{align}
        \max_{\J} \quad &\log_2 \det \left(\Id + \frac{P}{N\omega^2}\mqty[\G_{0,RS} & \G_{0,RD}] \J \mqty[\G_{0,RS}^\dagger \\ \G_{0,RD}^\dagger] \right) \label{eq:bound_obj}\\
        \text{s.t.} \quad &\Tr{\mqty[\alpha \I_S + \frac{1}{2\omega} \Asym(\G_{0,SS}) & \frac{1}{2\omega} \frac{1}{2i}\G_{0,SD} \\ \frac{1}{2\omega} \frac{-1}{2i}\G_{0,SD}^\dagger & \vb{0}] \J} \leq 1 \label{eq:tr_power_constraint}\\
        & \Tr{P \mqty[\vb{0} & - \G_{0,DS}^\dagger \Id_k \\ 
        \vb{0} & \U \Id_k] \J} = 0, \, \forall k  \label{eq:tr_re_sca_constraint}\\
        &\J \succeq \vb{0} 
    \end{align}
    \label{eq:direct_convex_bound}
\end{subequations}
This \textit{convex} optimization problem over the joint covariance \(\J \equiv \EX[\tj \tj^\dagger]/P\) is guaranteed to give a finite bound on the Shannon capacity given finite $\alpha, \Im\chi_D > 0$ (Appendix~\ref{asec:SDP_PD}) and can be solved using methods in semidefinite programming. 
Although (\ref{eq:direct_convex_bound}) is not strictly an SDP due to the log det objective, $-\log(\det(\cdot))$ is the standard logarithmic barrier function for SDP interior point methods and certain solvers support this objective natively (see for example SDPT3~\cite{Toh_SDPT3_2012}). In our numerical testing, we have found that using existing general convex optimization solvers for \eqref{eq:direct_convex_bound} is computationally expensive for even wavelength-scale domains, and instead we solve \eqref{eq:direct_convex_bound} using a ``dual water-filling'' approach that is a generalization of the method used in~\cite{ehrenborg_bounds_MIMO_2020} (see Appendix~\ref{asec:SDP_numerics}).  

\subsection{Insertion-Dominated Power Requirements and Singular-Value Decompositions} \label{sec:biconvex_relax}
Although the formulation above rigorously includes drive-power constraints on the Shannon capacity to all orders in scattering by the photonic structure, it is both practically and pedagogically useful to consider a relaxation of Eq.~\eqref{eq:Shannon_GaussianTO} that assumes drive-power considerations are dominated by insertion losses. 
This yields a solvable biconvex problem that elucidates fundamental principles of power transfer and channel allocation and may prove accurate in cases where structuring (local density of states enhancement) does not significantly influence extracted power, e.g., a laser incident on a scatterer or communication in an electrical wire. 
Mathematically, the contributions on the power constraint given by \(\G_{t,SS}\) are negligible compared to the insertion losses given by \(\sim \alpha \abs{\vj_i}^2\), which results in a current amplitude constraint of the form \(\abs{\vj_i}^2 \leq P / \alpha\).

We can simplify Eq.~\eqref{eq:Shannon_GaussianTO} by decomposing the Green's function via the singular value decomposition \(\G_{t,RS} = \sum_j \sigma_j \boldsymbol{\mu}_j \boldsymbol{\nu}_j^\dagger \equiv \boldsymbol{\mu} \Sigma \boldsymbol{\nu}^\dagger \).
We then \textit{define} each \(\boldsymbol{\mu}_j, \boldsymbol{\nu}_j\) as the available ``channels" for communication and \(\sigma_j\) as its corresponding channel strength~\cite{bucci_representation_1998,piestun_electromagnetic_2000,poon_degrees_2005,kuang_communication_channel_bounds_2022}. 
The compactness of \(\G\) ensures that inclusion of a finite number of channels 
\(n_c\) approximates the operator to high precision (i.e., the singular values decay)~\cite{hackbusch_hierarchical_matrices_2015}. 
This operator maps unit amplitude currents \(\boldsymbol{\nu}_j\) in \(S\) to fields \(\mathbf{e}_{R,j} = \frac{i}{\omega} \sigma_j \boldsymbol{\mu}_j\) in \(R\).
Diagonalizing \(\G_{t,RS}^\dagger \G_{t,RS} = \boldsymbol{\nu} \Sigma^2 \boldsymbol{\nu}^\dagger\), where \(\Sigma^2\) is a diagonal matrix of squared singular values, the Shannon capacity objective of Eq.~\eqref{eq:Shannon_GaussianTO} may then be written as~\cite{telatar_capacity_Gaussian_channels_1999} 
\begin{equation} \label{eq:shannon_objective_diag}
    \log_2 \det \left( \I + \dfrac{P}{N\omega^2} \Sigma \boldsymbol{\nu}^{\dagger} \J \boldsymbol{\nu} \Sigma \right).
\end{equation}
By Hadamard's inequality~\cite{różański_hadamard_2017}, this is maximized when \(\boldsymbol{\nu}^{\dagger} \J \boldsymbol{\nu}\) is diagonal, so we bound Eq.~\eqref{eq:shannon_objective_diag} by taking the degrees of freedom \(J_i\) to be the diagonal entries of \(\boldsymbol{\nu}^{\dagger} \J \boldsymbol{\nu}\).
The structural optimization problem~\eqref{eq:Shannon_GaussianTO} then takes the simpler form
\begin{equation} \label{eq:shannon_full_current}
\begin{aligned}
    \max_{\vb J, \epsilon(\vb r)} \quad & \sum_{i=1}^{n_c} \log_2 \left( 1 + \gamma J_i \abs{\sigma_i}^2 \right) \\
    \textrm{s.t.} \quad & \vb J \geq 0 \\ 
    & \vb 1^T \vb J \leq 1  \\ 
    & \left( \nabla \times \nabla \times - \epsilon(\vb r) \omega^2 \right)\G_t(\vb r, \vb r',\epsilon(\cdot);\omega) = \omega^2 \vb{I} \delta(\vb{r} - \vb{r'}) \\
    & \G_{t,RS}(\omega, \epsilon(\vb r)) = \sum_i^{n_c} \sigma_j \boldsymbol{\mu}_j \boldsymbol{\nu}_j^\dagger,
\end{aligned}
\end{equation}
where neglecting the work done by \(\vj_i\) allows writing the problem in terms of a \textit{single} free parameter \(\gamma = \frac{P}{\alpha \omega^2 N}\), and \(\vb J\) is the vector of current amplitudes \(J_i\). 
The utility of this formulation is now apparent: for problems where the impact of ``back-action" from the photonic environment on power requirements is negligible compared to losses in the drive currents, this relaxation provides a means to investigate Shannon capacity in terms of a ``scale-invariant" SNR parameter \(\gamma\).

For a known structure and therefore Green's function,~\eqref{eq:shannon_full_current} can be solved via the water-filling solution~\cite{tse_fundamentals_2005}.
The highly nonlinear dependence of the singular values \(\sigma_i\) on \(\epsilon(\vb r)\), however, makes this problem non-convex. Numerical solutions of~\eqref{eq:shannon_full_current} via gradient-based topology optimization, which relaxes the discrete optimization problem over \(\epsilon(\vb r)\) to a continuous counterpart~\cite{christiansen_inverse_2021}, can lead to useful designs, but do not provide bounds on the Shannon capacity over all possible structures. To obtain a bound on the Shannon capacity that is independent of geometry, one may relax~\eqref{eq:shannon_full_current} by ignoring the precise relation between singular values and the Green's function and instead encode structural information by imposing the same energy conservation constraints of Eq.~\eqref{eq:poynting} on the calculation of singular values. Thus, one may write
\begin{equation} \label{eq:shannon_relax_1} \begin{aligned}
    \underset{\mathbf{J}, \abs{\boldsymbol{\sigma}}^2}{\max} \qquad &  \sum_{i=1}^{n_c} \log_2 \left( 1 + \gamma J_i \abs{\sigma_i}^2 \right) \\
    \text{s.t.} \qquad & J_i \geq 0, \abs{\sigma_i}^2 \geq 0, \quad \forall i \leq n_c \\ 
    & \sum_i J_i \leq 1 \\ 
    & \sum_i \abs{\sigma_i}^2 \leq M_\mathtt{F} \\ 
	& \abs{\sigma_{\mathrm{max}}}^2 \leq M_1 
\end{aligned} \end{equation}
where \(M_\mathtt{F}, M_1\) are bounds on the square of the Frobenius norm \(\norm{\G_{t,RS}}_\mathtt{F}^2 \equiv \sum_i \abs{\sigma_i}^2\) and spectral norm \(\norm{\G_{t,RS}}_2^2 \equiv \abs{\sigma_{\textrm{max}}}^2 \) of the channel matrix, respectively.
Each of these bounds can be calculated in a structurally independent way while maintaining crucial information on the underlying physics of Maxwell's equations (see Appendix~\ref{asec:channel_bounds}). 
As long as constraints on \(\abs{\sigma_i}^2\) are bounds over all possible structures allowed in Eq.~\eqref{eq:shannon_full_current}, solutions to Eq.~\eqref{eq:shannon_relax_1} represent structure-agnostic limits on Shannon capacity between the \(S\) and \(R\) regions. 

Again, ignoring the limits \(M_\mathtt{F}\) and \(M_1\), Eq.~\eqref{eq:shannon_relax_1} is solved by the water-filling solution (see Appendix~\ref{asec:waterfilling}), leading to a cut-off in the number of utilized channels.
In full generality, the objective is biconvex and therefore difficult to solve~\cite{floudas_optimization_2000}.
While numerical methods to find global optima of biconvex functions exist (see for example~\cite{floudas_optimization_2000}),  existing implementations do not easily allow for this objective function (see for example the cGOP package). Development and application of techniques to solve for the global optimum of~\eqref{eq:shannon_relax_1} is a promising avenue of future work. 
Surprisingly, however, incorporating the Frobenius norm bound \(M_\mathtt{F}\) but excluding the maximum singular value bound \(M_1\) makes ~\eqref{eq:shannon_relax_1} amenable to closed-form solutions. As shown in Appendix~\ref{asec:shannon_relax_proof1}, the optimal solution to this problem may be written as follows:  
\begin{gather} \label{eq:general_sol_shannon}
C(M_F) = f_{n=n_{\mathrm{opt}}} = n_{\mathrm{opt}}\log_2 \left(1 + \gamma'/ n_{\mathrm{opt}}^2 \right) \\ n_{\mathrm{opt}} = 
\begin{cases}
     \floor{r\sqrt{\gamma'}} \textrm{ or } \ceil{r\sqrt{\gamma'}}, \quad & \floor{r\sqrt{\gamma'}} \leq n_c \\ 
     n_c, \quad & \floor{r\sqrt{\gamma'}} \geq n_c
\end{cases}
\end{gather}
where \(n_{\mathrm{opt}} \geq 1\), \(r = \max_{x\in [0, 1]} \left[ x \log_2 \left(1+1/x^2\right) \right] \approx 0.505 \), and we have introduced the \textit{total} signal-to-noise ratio \(\gamma' = \gamma M_\mathtt{F}\); note that the floor and ceiling functions are a consequence of the number of populated channels being an integer allocation.
This solution corresponds to evenly allocating capacity and power among a subset of all channels.
Notably, \(r\sqrt{\gamma'} \leq 1 \implies n_{\mathrm{opt}} = 1\) and \(r\sqrt{\gamma'} \geq n_c \implies n_{\mathrm{opt}} = n_c\). 
In the high-noise regime (i.e., SNR dominated by increasing channel capacities), the Shannon capacity is maximized by allocating all available channel capacity and current to a single channel. 
As noise is lowered, the number of channels increases by integer amounts until the low-noise regime (i.e., SNR dominated by low noise, regardless of channel capacities) is reached, in which case every channel has capacity \(M_\mathtt{F}/n_c\) and power \(P/n_c\).
In both cases, current allocations naturally satisfy the water-filling solution given their corresponding channel capacity allocation. 
As shown in Appendix~\ref{asec:shannon_relax_proof1}, asymptotic solutions in the low and high SNR regimes take the form
\begin{equation} \begin{aligned}
    C(M_\mathtt{F}) &\xrightarrow[]{\gamma' \ll 1} \gamma' / \log 2, \\ 
    C(M_\mathtt{F}) &\xrightarrow[]{\gamma' \gg 1} n_c \log_2 (\gamma' / n_c^2),
\end{aligned} \end{equation} 
in agreement with the analysis above. 
Intuitively, this solution is unrealistic in the low-noise regime: 
the compactness of \(\G\) implies that singular values must decay, making even allocation of singular values among all channels impractical.
Regardless, this method provides a simple way to calculate limits on Shannon capacity given only a sum-rule on the Frobenius norm \(M_\mathtt{F}\) of the channel matrix.
In the high-noise and low-noise regimes, it is straightforward to add the maximum singular value bound \(M_1\) to this analysis; this solution is denoted \(C(M_\mathtt{F}, M_1)\).
If \(r\sqrt{\gamma'} \leq 1\) (the regime where \(n_{\mathrm{opt}} = 1\)) a simple substitution of \(M_\mathtt{F} \to M_1\) yields a bound. 
In general, if \(M_\mathtt{F}/n_{\mathrm{opt}} \leq M_1\), the maximum singular value constraint is not active in the solution \(C(M_\mathtt{F}, M_1)\) and so it coincides with \(C(M_\mathtt{F})\). The amenability of this rigorous formulation to analysis via the familiar language of channel capacity and power allocation makes this conceptual ground for further studies of the impact of structuring and SNR on optimal channel design.

\subsection{Relation to Prior Work}

The challenge of~\eqref{eq:Shannon_GaussianTO} from an optimization perspective is the structural optimization over $\epsilon(\vb r)$ which is both high-dimensional and non-convex. The main novelty of our results is the ability to account for the deliberate engineering of the propagation environment through such structural optimization, which is increasingly relevant given the rise of integrated photonics platforms. Mathematically, prior work on bounding structural optimization in photonics focused predominantly on objectives that are quadratic functions of the fields~\cite{chao_physical_2022,angeris_heuristic_2021}; here we bound the Shannon capacity $\log\det$ objective using a relaxation with a probability distribution based interpretation that ties directly into the probabilistic proof of Shannon's noisy channel coding theorem~\cite{shannon_mathematical_1948,cover_information_2006,mackay_inference_2019}. 

If the structure \(\epsilon(\vb r)\) and therefore the channel matrix \(\G_{t,RS}\) is fixed, the Shannon capacity is given by a convex optimization problem that can be solved via the water-filling solution. Prior work on Shannon capacity in electromagnetism have generally made such a simplifying assumption, either taking the propagation environment to be fully vacuum~\cite{miller_communicating_2000,miller_waves_2019,kuang_communication_channel_bounds_2022} or considering a fixed antenna structure in free space~\cite{ehrenborg_bounds_MIMO_2017,ehrenborg_bounds_MIMO_2020,ehrenborg_capacity_bounds_2021}. 
Specifically, Refs.~\cite{ehrenborg_bounds_MIMO_2017,ehrenborg_bounds_MIMO_2020,ehrenborg_capacity_bounds_2021} investigated a planar antenna sheet surrounded by an idealized spherical shell receiver that captures all waves propagating through free space. 
Ref.~\cite{kuang_communication_channel_bounds_2022} calculated the Shannon capacity in vacuum between concentric sphere-shell sender and receiver regions: due to the domain monotonicity of the singular values of $\G_0$~\cite{molesky_mathbbt-operator_2019}, this bounds the capacity between smaller sender-receiver regions completely enclosed within the larger sphere-shell. 
Related works~\cite{miller_communicating_2000, miller_waves_2019} have established sum-rule bounds on these vacuum singular values, establishing a qualitative connection between information transfer and the notion of maximizing field intensity at the receiver. 

There are also differences in the exact constraint used to restrict the input current magnitude: Ref.~\cite{kuang_communication_channel_bounds_2022}~implicitly used a Frobenius norm constraint on the input current covariance, which is similar to the insertion-loss dominated case of~\ref{sec:biconvex_relax}. We explicitly impose a constraint on the total power required to maintain the input currents, including not just insertion loss but also the power going into propagating fields, similar to Ref.~\cite{ehrenborg_bounds_MIMO_2017,ehrenborg_bounds_MIMO_2020,ehrenborg_capacity_bounds_2021}. Refs.~\cite{ehrenborg_bounds_MIMO_2017,ehrenborg_bounds_MIMO_2020} also considered further restrictions on the antenna efficiency, defined as the ratio of radiated power to total power; Ref.~\cite{ehrenborg_capacity_bounds_2021} included constraints on the ratio of stored vs. radiated power in the antenna as a proxy for the operating bandwidth. 
Our bound formulation \eqref{eq:direct_convex_bound} given in the previous section is flexible enough to accommodate such additional constraints; given the conceptual richness that structural optimization adds to the problem, we have opted to only include the drive power constraint \eqref{eq:drive_power_constraint} in this manuscript for simplicity. 

\section{Illustrative Examples} \label{sec:numerics}

\begin{figure*}[t]
    \centering
    \includegraphics[width=\textwidth]{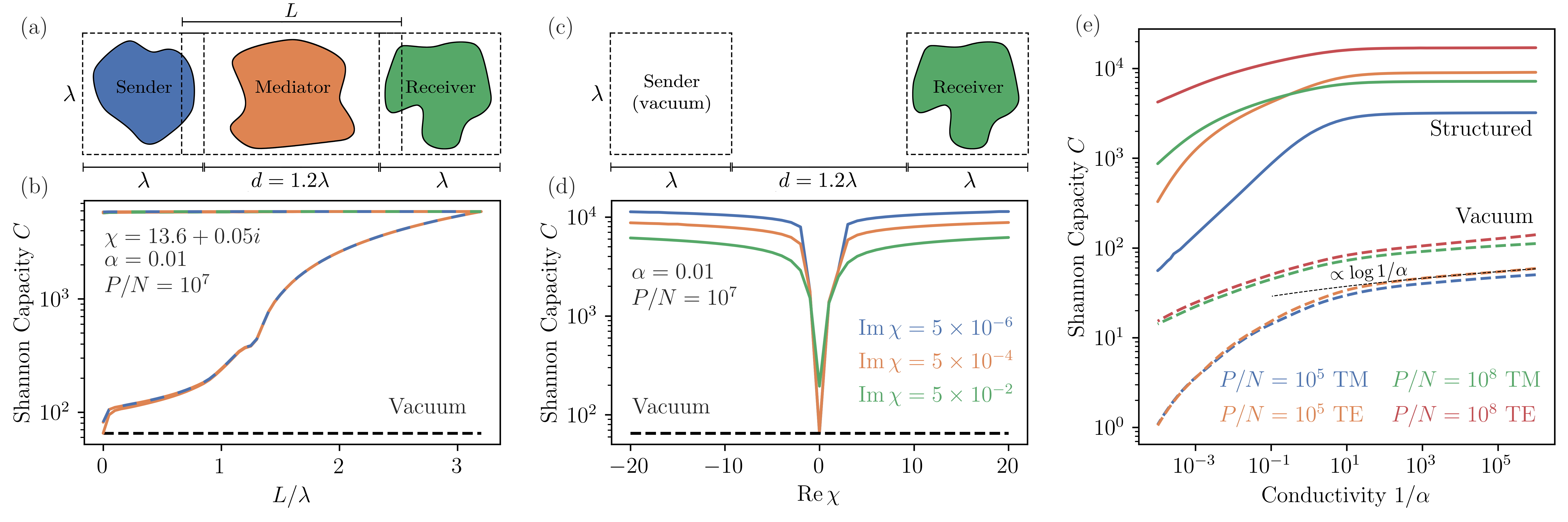}
    \caption{
    Upper bounds on the Shannon capacity $C$ computed via the convex relaxation of Eq.~\eqref{eq:direct_convex_bound}. (a) Schematic of geometrical configuration pertaining to (b), showing sender and receiver regions of size \(\lambda \times \lambda\) separated by a surface-surface distance \(1.2\lambda\) by a mediator region of variable size \(L \times \lambda\). 
    (b) Bounds on $C$ as a function of the mediator length \(L\) for different structuring combinations of sender, mediator, and/or receiver regions.
    Results highlight the out-sized impact of structuring the receiver region in this regime of extraction-dominated power. 
    Zebra-coloring in (b) corresponds to designing the corresponding regions as colored in (a) (e.g., orange-blue-green zebra coloring pertains to the scenario where all three regions are designable). 
    (c) Schematic pertaining to (d) and (e), where now only the receiver region is designable. 
    (d) Bounds as a function of \(\Re \chi\) corresponding to structuring only the receiver region, corroborating the importance of engineering field enhancements at the receiver.
    (e) Limits on $C$ for TE and TM polarizations (solid lines) and associated vacuum values (dashed lines) as a function of \(1/\alpha\) consisting of a structured receiver and vacuum sender of sizes \(\lambda \times \lambda\) a distance \(1.2\lambda\) apart. 
    Results showcase a logarithmic divergence of these limits with decreasing $\alpha$ stemming from the increased contribution of evanescent ``dark currents" to information transfer and potential enhancements to Shannon capacity by more than an order of magnitude via structuring the receiver.
    Note that \(\alpha\) has units of resistivity \(Z \lambda\) and \(P/N\) has units of \(\epsilon_0 c \lambda\). 
    }
    \label{fig:convex_combined}
\end{figure*}

\begin{figure}
    \centering
    \includegraphics[width=1\linewidth]{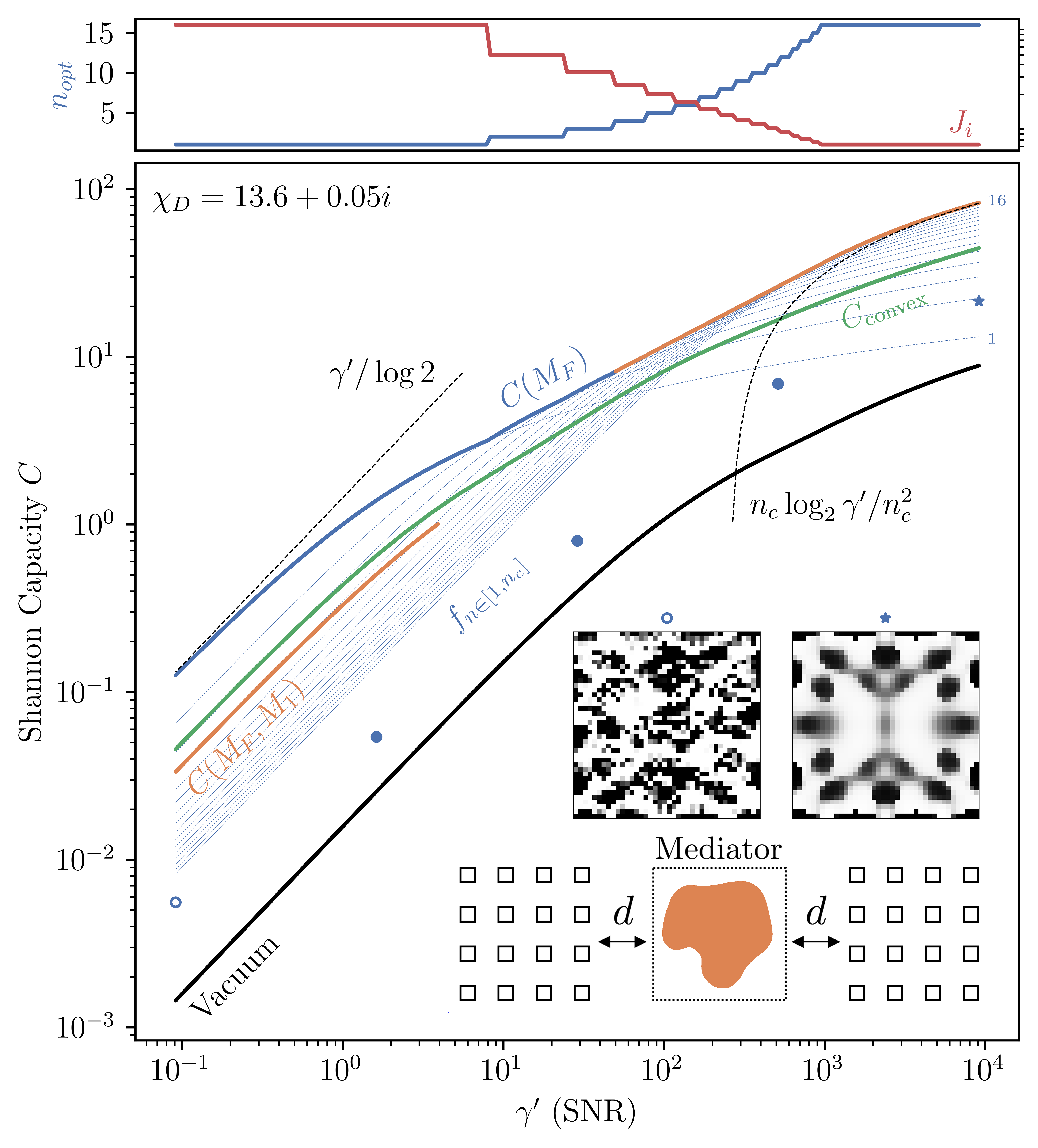}
    \caption{Shannon capacity limits (solid colored lines), inverse designs (dots), and associated vacuum performance (black line) in the insertion-dominated power regime as a function of total signal-to-noise ratio (SNR) \(\gamma'=\frac{P}{\alpha\omega^2 N} M_\mathrm{F}\) for \(4 \times 4\) source and receiver arrays spanning a space of size \(\lambda \times \lambda\) connected by a designable mediator region of susceptibility \(\chi = 13.6+0.05i\). 
    As described in the main text, \(C(M_\mathtt{F})\) is an analytic solution to the bounding problem subject to a Frobenius norm constraint $\sum_i |\sigma_i|^2 \leq M_\mathrm{F}$, while \(C(M_\mathtt{F}, M_1)\) incorporates an additional constraint on the maximum singular channel capacity $|\sigma_\mathrm{max}| \leq M_1$. 
    In order to show how \(C(M_\mathtt{F})\) is a pointwise maximum of uniform allocation solutions \(f_n\) (Eq.~\eqref{eq:general_sol_shannon}), \(f_n\) is plotted (dashed blue) for all \(n\in[1,n_c=16]\).
    Analytically derived asymptotic solutions (dashed black lines) incorporating trace constraints are also shown to converge to their respective analytic solution within this range of SNR. 
    Bounds \(C_\textrm{convex}\) obtained via the direct convex relaxation of Eq.~\eqref{eq:direct_convex_bound} are also computed in the insertion-dominated limit of large \(\alpha \gg \norm{\Asym{\G_{0,SS}}}/2\omega\). Two representative inverse designs for low (left inset) and high (right inset) SNR regimes are shown on the plot. 
    Notably, their performance is within an order of magnitude of the tightest bound in each regime. For this problem, \(M_1 \approx M_\mathtt{F}/4\).
    }
    \label{fig:array_results}
\end{figure}

In the prior section, we derived two methods to compute upper bounds on the Shannon capacity between sender and receiver regions containing any possible photonic structure in either region or their surroundings. 
As described in Section~\ref{sec:theory}, the total drive power constraint Eq.~\ref{eq:shannon_objective_power} can be considered in two regimes. 
In one regime, handled by the direct convex relaxation, the extracted power (determined by Maxwell's equations) and the insertion power (as abstracted by a resistivity associated with driving the initial currents) are comparable: 
\(\frac{1}{2\omega} \norm{\Asym \G_{t,SS}}_2^2 \sim \alpha \).
In the other, treated by both the direct convex relaxation and the simpler biconvex problem, the cost of driving initial currents dominates power consumption (i.e., insertion-dominated): \(\frac{1}{2\omega} \norm{\Asym \G_{t,SS}}_2^2 \ll \alpha \).
In this section, we evaluate bounds in both of these contexts for several illustrative 2D configurations, showcasing the possibility of enhancing information transfer via photonic design and also revealing scaling characteristics with respect to material and design parameters. 
We compute bounds numerically both for out-of-plane (scalar electric field) polarization (TM) and in-plane (vector electric field) polarization (TE) using a finite-difference frequency domain (FDFD) discretization; see Appendix~\ref{asec:discretization} for more details on evaluating the Shannon capacity over a discrete grid. 
In the particular case of insertion-dominated power, we also compare the bounds with structures obtained via topology optimization, using an in-house automatic differentiable FDFD Maxwell solver written in JAX~\cite{jax_github_2018, optax} to compute gradients of the Shannon capacity with respect to structural degrees of freedom. 
Topology optimization of Shannon capacity is challenging due to its dependence on the total Green's function; more efficient methods to compute these structures is the subject of future work.
    
\emph{Direct Convex Relaxation:} We consider a scenario, shown schematically in Fig.~\ref{fig:convex_combined}(a), involving sender (\(S\)) and receiver (\(R\)) square regions of size $\lambda \times \lambda$ with centers aligned along the $x$-axis and separated by a surface--surface distance of $1.2\lambda$. A rectangular mediator region \(M\) of length $L$ and width $\lambda$ (along the $y$ dimension) is placed between \(S\) and \(R\). Bounds are computed by numerically solving (\ref{eq:direct_convex_bound}) with energy conservation imposed globally over the entire design region $I_D$, i.e. a single projection domain/constraint.

To investigate the impact of structuring different regions on the Shannon capacity, Fig.~\ref{fig:convex_combined}(b) shows bounds for TM polarized sources, assuming a fixed impedance $\alpha=0.01 Z\lambda$ and SNR \(P/N=10^7 \epsilon_0 c \lambda\), as a function of $L$ and for different combinations of designable $S$, $M$, and $R$ regions. As $L$ increases, the mediator eventually overlaps with and completely encloses the \(S\) and \(R\) regions. 
The results reveal a striking feature: in contrast to designing the sender, structuring the receiver has by far the greatest impact on the Shannon capacity. 
Intuitively, designing \(R\) allows large field concentrations \(\ve_R\). Such density of states enhancements are further corroborated by the dependence of the bounds on material susceptibility $\chi$ seen in Fig.~\ref{fig:convex_combined}(d) corresponding to only structuring the receiver region (Fig~\ref{fig:convex_combined}(c)). In particular, bounds grow rapidly with $\abs{\Re\chi}$ for weak materials before saturating, and exhibit highly subdued scaling with respect to material loss (\(\Im \chi\)) consistent with resonant enhancements also seen in prior works on photonic limits~\cite{venkataram_fundamental_2020-1,molesky_hierarchical_2020}: decreasing $\Im\chi$ by two orders of magnitude yields an increase in $C$ of less than a factor of two. The exact nature of this scaling behavior is the subject of future work.

The apparent ineffectiveness of structuring the sender to increase $C$ stems from the implicit freedom associated with designing the input drive, which allows for arbitrary input currents $\vb{j}_i$ within the $S$ region, capped in magnitude only by the finite drive power constraint. In particular, the field $\vb{e}_R$ produced by a free current source $\vb{j}_i$ in a structured $S$ region containing a polarization current $\vb{j}_s$ is in fact equivalent to one produced by a corresponding ``dressed" input current $\vb{j}_t = \vb{j}_i + \vb{j}_s$ in vacuum. 
The imposition of a finite drive power requirement with an insertion power contribution does, however, create a distinction between these two scenarios. In the $\alpha \rightarrow 0$ limit where there is no power cost to modifying the input current $\vb{j}_i$, there is no advantage in structuring since non-zero $\vb{j}_s$ incurs further material losses compared to inserting $\vb{j}_t$ directly in vacuum. Conversely, in the insertion dominated regime $\alpha \gg \norm{\Asym{\G_{t,SS}}}_2/2\omega$ increasing $\vb{j}_i$ has a steep cost, making it more advantageous to achieve larger field enhancements by increasing $\vb{j}_t$ via structuring. Clearly, the relative importance of insertion- versus extraction-dominated power as determined by $\alpha$ plays a central role in determining how $C$ should be optimized.

Figure~\ref{fig:convex_combined}(e) shows bounds on \(C\) as a function of \(1/\alpha\) for both TM and TE sources, allowing for structuring only in \(R\) but fixing the distance between the \(S\) and \(R\) regions to be $1.2\lambda$ (see inset). In the extraction dominated regime ($\alpha \rightarrow 0$), the Shannon capacity exhibits a logarithmic dependence on $\alpha$ attributable to ``dark currents'' which do not radiate into the far field as input signals. 
Mathematically, these dark currents lie within the nullspace of $\Asym\G_{SS}$, so the extracted power contribution $-\frac{1}{2\omega}\Re{\vb{j}_i^\dagger \G \vb{j}_i}$ to the drive power constraint~\eqref{eq:drive_power_constraint} is $0$.  While dark currents do not radiate nor consume power, they nonetheless produce evanescent fields that are picked up by a receiver at a finite distance, allowing for information transfer. 
The dark current amplitudes are limited by $\alpha$, so for fixed $P/N$ their contribution to the covariance $\J$ is inversely proportional to $\alpha$, resulting in a logarithmic dependence of $C$ on $1/\alpha$ via the $\log_2\det$ objective.  This leads to a logarithmic divergence in the Shannon capacity as \(\alpha \to 0\), even for finite power. As $\alpha$ increases, $C$ decreases due to power restrictions on input currents, eventually reaching the insertion dominated regime ($\alpha \to \infty$) where the extracted power is dwarfed by the insertion impedance term $\alpha \abs{\vb{j}_i}^2$. 
Neglecting the extracted power, the resulting power constraint on the input currents \(\sim \abs{\vj_i}^2 \leq P/\alpha\) leads to a problem that can be bounded via both the direct convex relaxation and the biconvex relaxation of \eqref{eq:shannon_relax_1}, as illustrated below. 

\emph{Insertion-Dominated Power:} 
In this section we study the insertion dominated regime (\(\alpha \gg \frac{1}{2\omega} \norm{\Asym \G_{t,SS}}\)), where one can also apply the biconvex relaxation given in Eq.~\eqref{eq:shannon_relax_1}. 
As described in Section~\ref{sec:biconvex_relax}, this simplification permits a definition of channels as the singular vectors of \(\G_{t,RS}\), with the relevant physics encoded via sum-rules on the associated singular values. This leads to a simplified, biconvex optimization problem~\eqref{eq:shannon_relax_1} with two analytic solutions: \(C(M_\mathtt{F})\) maximizes Shannon capacity subject to a bound $M_\text{F}$ on the maximum sum of all channel capacities (Frobenius norm), and \(C(M_\mathtt{F}, M_1)\) additionally incorporates a bound on the spectral norm \(M_1\) or the maximum channel capacity of the Green's function. The resulting bounds depend only on the single universal total signal-to-noise ratio parameter \(\gamma' = \frac{P}{\alpha \omega^2 N} \times M_\text{F}\).

We consider a simplified scenario shown schematically in Fig.~\ref{fig:array_results}, consisting of discrete \(S\) and \(R\) regions comprising $4\times 4$ vacuum pixel arrays separated by a designable contiguous $M$ mediator of material $\chi$.  Hence, there are \(16\) orthogonal singular vector channels available for communication. 
The figure shows bounds \(C(M_\mathtt{F})\) and \(C(M_\mathtt{F}, M_1)\) alongside comparable inverse designs. 
Notably, the tightest bounds are found to be within or near an order of magnitude of achievable performance for all studied noise regimes, with particularly realistic performance for high SNR.
The additional tightness in \(C(M_\mathtt{F}, M_1)\) compared to \(C(M_\mathtt{F})\) for high-noise regimes motivates extending the domain of this solution via further analysis. 
The corresponding channel occupation \(n_{\mathrm{opt}}\) and current amplitude allocation \(J_i = 1/n_{\mathrm{opt}}\) for \(C(M_\mathtt{F})\) (shown above the main plot) demonstrate the nonlinear dependence of these two values on noise. 
Even allocation \(f_{n\in[1,16]}\) are also shown, demonstrating how the solution \(C(M_\mathtt{F})\) is a pointwise maximum of these uniform allocations. 
As discussed earlier, higher SNR pushes allocation of \(J_i\) and channel capacity to a single channel leading to logarithmic dependence of the Shannon capacity on $\gamma'$, while low SNR pushes even allocations of \(J_i\) and channel capacity across all channels, leading instead to linear scaling.  This is consistent with the results presented in Fig.~\ref{fig:convex_combined}, where large and small \(1/\alpha\) result in logarithmic and linear scaling of Shannon capacity, respectively. The biconvex formulation, however, suggests that while it may not be physically possible to optimize for any desired channel matrix, the SNR determines a target number of high-capacity channels. 
Accordingly, high SNR structures display higher symmetry than the seemingly random structure at low SNR owing to maximizing information transfer among one or many channels, respectively.
Finally, the figure also shows bounds \(C_\textrm{convex}\) pertaining to the convex relaxation of Eq.~\eqref{eq:direct_convex_bound} and obtained by enforcing large \(\alpha\) and \(P\) so as to reach the insertion-dominated power regime. The latter is found to be looser than \(C(M_\mathtt{F}, M_1)\) in the low SNR regime and tighter in the high SNR regime, showcasing the versatility of the convex formulation and, despite its complexity, its general agreement with the conceptually simpler biconvex relaxation.

\section{Concluding Remarks}

In this article, we formalized an appropriate figure of merit for maximizing electromagnetic communication subject to arbitrary photonic structuring, and connected it to existing work restricted to vacuum propagation and/or fixed antenna geometries. 
The proposed framework rigorously connects the mathematics of electromagnetic wave propagation in photonic media to key information-theoretic quantities such as the mutual information and the Shannon capacity. 
We derived two possible relaxations of this problem (one valid in general settings and the other in regimes where insertion losses dominate) that encode wave physics from Maxwell's equations along with drive-power considerations, to arrive at structure-agnostic upper bounds on Shannon capacity. 

These bounds not only yield quantitative performance targets, but also critical scaling information with respect to device size, SNR, and material strength. In particular, our numerical evaluation of illustrative examples has highlighted the importance of optimizing the receiver, with more than an order of magnitude potential improvement of the Shannon capacity. The bounds have weak scaling with material loss, a fact with both theoretical and practical significance and likely related to similar sub-linear (logarithmic) scaling with loss seen in prior works involving distributed sources such as in near-field radiative heat transfer~\cite{venkataram_fundamental_2020-1}. 
When radiated power costs are dominant, the Shannon capacity displays logarithmic growth as $\alpha \rightarrow 0$ driven by nonradiative dark currents. 
As \(\alpha\) is increased and insertion impedance restricts input currents, we observe capacity scaling with SNR transitioning from linear for low SNR to logarithmic for high SNR, alongside a gradual increase in the number of optimal sub-channels used. Furthermore, topology optimized structures are shown to generally achieve performance within an order of magnitude of our bounds across a wide range of SNR. 

Looking ahead, we believe there is significant promise for further research into photonics design for information transfer. 
The results in this manuscript were computed in 2D and serve mainly as a proof of concept; we anticipate that our formulation may be adapted to study the information capacity of practical devices of increasing interest such as on-chip optical communication and image processing using integrated photonics. 
Doing so would also require improving computational techniques for evaluating the bounds at larger scale and in 3D, along with the development of new inverse design schemes specifically for optimizing information transfer. 
Our current results are also derived at a single frequency; further work remains to be done on generalizing the bounds for finite bandwidths, perhaps in the spirit of spectral sum rules~\cite{barnett_sum_1996,scheel_sum_2008,zhang_all_2023} or delay-bandwidth products~\cite{tucker_slow-light_2005,miller_fundamental_2007}. Another promising line of study is to incorporate macroscopic quantum electrodynamics~\cite{scheel_macroscopicQED_2009} into the formulation for investigating information transfer in quantum communication systems.  

\textbf{Funding: } we acknowledge the support by a Princeton SEAS Innovation Grant and by the Cornell Center for Materials Research (MRSEC). 
S.M. acknowledges financial support from the Canada First Research Excellence Fund via the Institut de Valorisation des Données (IVADO) collaboration.
The simulations presented in this article were performed on computational resources managed and supported by Princeton Research Computing, a consortium of groups including the Princeton Institute for Computational Science and Engineering (PICSciE) and the Office of Information Technology's High Performance Computing Center and Visualization Laboratory at Princeton University.

\appendix
\raggedbottom   

\section{Water-filling Solution} \label{asec:waterfilling}
As stated in the main text, the convex problem 
\begin{equation} \label{eq:simple_waterfilling}
    \max_{\vb{J} \succ \vb{0}, \Tr(\vb{J})\leq P} \quad \log_2 \det(\mI + \mH\vb{J}\mH^\dagger)
\end{equation}
can be solved via the water-filling solution. 
In this appendix, this solution is detailed. 
As in the main text, taking the singular value decomposition of \(\mH = \vb U\Sigma \vb V^\dagger\) gives the objective \(\log_2 \det (\mI + \Sigma \vb J' \Sigma )\), where \(\vb J' = \vb V^\dagger \vb J \vb V\). 
By Hadamard's inequality~\cite{różański_hadamard_2017}, this is maximized when \(\vb J'\) is diagonal, and we label these diagonal entries \(J_i\). 
The problem becomes
\begin{equation}
    \max_{J_i \geq 0, \sum J_i \leq P} \quad \log_2 \prod (1 + \sigma_i^2 J_i)
\end{equation}
where \(\sigma_i\) are the diagonal entries of \(\Sigma\). 
This equation is analogous to Eq.~\ref{eq:shannon_relax_1}.
The Lagrangian and its derivative can then be written 
\begin{equation}\begin{gathered} 
    \mathcal{L} = \sum \log(1+\sigma_i^2 J_i) + \lambda(P-\sum J_i) + \sum \mu_i J_i \\ 
    \frac{\partial \mathcal{L}}{\partial J_i} = \dfrac{\sigma_i^2}{1+\sigma_i^2 J_i} - \lambda + \mu_i = 0
\end{gathered}\end{equation}
where \(\lambda, \mu_i\) are the Lagrange multipliers. Complementary slackness and primal feasibility guarantee that at the optimal point,  \(\mu_i J_i = 0\) for all \(i\) and \(\sum J_i = P\). If \(\mu_i = 0\), then \(J_i = \frac{1}{\lambda} - \frac{1}{\sigma_i^2}\). Otherwise, \(J_i = 0\) and \(\mu_i = \lambda - \sigma_i^2\). Thus, the allocations \(J_i\) can be chosen to be 
\begin{equation}
    J_i = \max \left\{ 0,\lambda^{-1} - \sigma_i^{-2} \right\}
\end{equation}
where the value of \(\lambda\) is determined by the primal feasibility condition \(\sum \max \{ 0, \lambda^{-1} - \sigma_i^{-2} \} = P\). 
This solution can be interpreted as ``filling" channels of largest \(\sigma_i\), leaving large \(\sigma_i^{-2}\) unallocated, until all the ``water" (power) is consumed. 
An immediate consequence is that a finite number of ``channels" with capacities \(\sigma_i\) are utilized, and therefore for finite power there is always a finite-rank approximation of \(\vb H\) that is sufficient to model the communication problem. 

A weighted water-filling problem 
\begin{equation}
    \max_{\vb{J} \succ \vb{0}, \Tr(\vb{BJ})\leq P} \quad \log_2 \det(\mI + \mH\vb{J}\mH^\dagger)
\end{equation}
for a known positive-definite matrix \(\vb B\) can be solved by performing any decomposition \(\vb B = \vb U^\dagger \vb U\) (e.g., square root, Cholesky decomposition). The transformed problem takes the form 
\begin{equation}
    \max_{\vb{\widetilde{J}} \succ \vb{0}, \Tr(\vb{\widetilde{J}})\leq P} \quad \log_2 \det(\mI + \widetilde{\mH}\vb{\widetilde{J}}\widetilde{\mH}^\dagger)
\end{equation}
with \(\widetilde{\mH} = \mH \vb U^{-1}\) and \(\widetilde{\vb J} = \vb{UJU^\dagger}\), giving a problem of the form of Eq.~\eqref{eq:simple_waterfilling}.

\section{Convergence of Photonic Shannon Capacity with Regards to Discretization\label{asec:discretization}}
In principle, the photonic Shannon capacity involves input and output signals that can be infinite-dimensional: currents and fields over continuous spatial regions. In order to numerically solve Shannon capacity maximization problems such as (\ref{eq:Shannon_GaussianTO}), discretization of the continuous problem is necessary. This section discusses how the spatial resolution $\Delta$ of discretization schemes such as finite-difference frequency domain (FDFD) or the finite element method (FEM) enters into photonic Shannon capacity calculations, and demonstrates that the photonic Shannon capacity between the $S$ and $R$ region for any given structure $\epsilon(\vb{r})$ converges in the infinite resolution limit $\Delta \rightarrow 0$. 

Consider the covariance maximization problem for evaluating Shannon capacity
\begin{subequations}
    \begin{align}
        \max_{\J} \quad & \log_2 \det(\Id + \dfrac{P}{N\omega^2} \G_{t, RS}(\epsilon)\J \G_{t, RS}^\dagger (\epsilon)), \label{eq:app_shannon_objective} \\
        \text{s.t.} \quad & \Tr( \J) \leq 1, \label{eq:app_trCstrt}\\
        & \J \succeq \vb{0},
    \end{align}
    \label{eq:app_ShannonC}
\end{subequations}
where for simplicity we have taken the insertion-loss dominated power constraint and set $\alpha = 1$. $\J$ is the power-scaled covariance of the input currents \(\J(\vb{x},\vb{y}) = \EX[\vj_i(\vb{x}) \vj_i^\dagger(\vb{y})]/P\), and the output signal are the electric fields $\vb{e}_R = \frac{i}{\omega} \G_{t,RS} \vj_i + \vb{n}$, where $\vb{n}(\vb{r})$ is spatially distributed complex Gaussian white noise with $\EX[\vb{n}(\vb{r})] = 0$ and auto-covariance $\EX[\vb{n}(\vb{x})\vb{n}(\vb{y})^\dagger] = N \hat{\vb{I}} \delta(\vb{x}-\vb{y})$. 

Upon discretization, the continuous fields become discrete vectors and the operators become finite size matrices, with factors of $\Delta$ included in expressions as appropriate. One place where a resolution factor is needed is the trace constraint (\ref{eq:app_trCstrt}), which comes from the cyclic permutation of $\int_S |\vb{j}(\vb{r})|^2 \,\dd \vb{r} = P$. Discretization converts this to a numerical quadrature expression which involves multiplication by the voxel size $\Delta$.

Another more subtle scaling with $\Delta$ is the covariance of the discretized noise vector $\vb{n}$ which satisfies $\EX[\vb{n}_i \vb{n}_j^\dagger] = (N/\Delta) \vb{\hat{I}} \delta_{ij}$ where $\delta_{ij}$ is the Kronecker delta. Without going through detailed stochastic analysis~\cite{grimmett_probability_2001}, this inverse $\Delta$ scaling can be understood intuitively from the perspective that the entries of a discretized vector represent a voxel average of the continuous field. If we halve the resolution and use voxels of size $2\Delta$, the discretized noise over the larger voxel can be viewed as the average of the discretized noise of 2 $\Delta$-sized voxels, and averaging 2 independent identically distributed random variables reduces the variance by a factor of $1/2$. The limit $\Delta \rightarrow 0$ recovers the delta function auto-covariance of the continuous spatial white noise. 

A discretized version of (\ref{eq:app_ShannonC}) is thus given by
\begin{align*}
    \max_{\vb{J}} \quad & \log_2 \det(\vb{I} + \dfrac{P}{(N/\Delta) \omega^2} \vb{G}_{t, RS}(\epsilon) \vb{J} \vb{G}_{t, RS}^\dagger (\epsilon)), \\
        \text{s.t.} \quad & \Tr( \vb{J}) \Delta \leq 1, \\
        & \vb{J} \succeq \vb{0},
\end{align*}
where we have used regular boldface to denote finite matrices. If we now further absorb a factor of $\Delta$ into $\vb{J}$, we see that the resolution scaling in the noise and the trace cancel out:
\begin{subequations}
    \begin{align} \label{eq:app_resFree_ShannonC}
    \max_{\vb{J}} \quad & \log_2 \det(\vb{I} + \dfrac{P}{N/ \omega^2} \vb{G}_{t, RS}(\epsilon) \vb{J} \vb{G}_{t, RS}^\dagger (\epsilon)), \\
        \text{s.t.} \quad & \Tr( \vb{J}) \leq 1, \\
        & \vb{J} \succeq \vb{0}.
\end{align}
\end{subequations}
This is a convex optimization over the finite sized matrix $\vb{J}$ independent of any explicit resolution factors $\Delta$, which can be solved by setting $\vb{J}$ as diagonal in the SVD basis of $\vb{G}_{t,RS}$ and waterfilling, as described in Appendix~\ref{asec:waterfilling}; the result solely depends on the larger singular values of $\vb{G}_{t,RS}$. (\ref{eq:app_resFree_ShannonC}) thus converges to the photonic Shannon capacity of the continuous structure and fields as $\Delta \rightarrow 0$ just as the larger singular values of $\vb{G}_{t,RS}$ converge to those of $\G_{t,RS}$. Given that $\G_{t,RS}$ is low-rank when $R$ and $S$ do not overlap~\cite{hackbusch_hierarchical_matrices_2015}, there are only finitely many Green's function singular values that need to converge; this happens when $\Delta$ is small compared to the characteristic length-scales of the corresponding singular vectors. 

\section{Proof that the Convex Relaxation Yields a Finite Bound} \label{asec:SDP_PD}

The convex relaxation approach outlined in Section~\ref{sec:convex_relax} gives a bound written as the solution to maximizing an objective function of the form  $\log\det(\mI + \mH \mJ \mH^\dagger)$ subject to the PSD constraint $\mJ \succeq \vb{0}$ and various linear trace constraints on $\mJ$. 
Upon first sight, it may be unclear whether such an optimization has a finite maximum: one can imagine that the eigenvalues of $\mH\mJ\mH^\dagger \succeq \vb{0}$ may increase without bound and lead to a diverging $\log\det$, and indeed this is what happens without any linear trace constraints. 
The key for a finite maximum of \eqref{eq:direct_convex_bound} is the existence of a linear trace constraint of the form $\Tr[\vb{A} \mJ] \leq P_A$ where $\vb{A} \succ 0$ and $P_A>0$, i.e., the program
\begin{subequations}
    \begin{align}
        \max_{\mJ} \quad &\log\det(\mI + \mH\mJ\mH^\dagger) \\
        \text{s.t.} \quad & \Tr[\vb{A} \mJ] \leq P_A \\
        & \mJ \succeq 0
    \end{align}
    \label{eq:PD_waterfilling}
\end{subequations}
has a finite maximum. 
To see this, observe that given strictly definite $\vb{A}$, it must be the case that $\Tr[\vb{A}\mJ]>0$ for all positive semidefinite $\mJ \neq \vb{0}$. 
Therefore, in the vector space of Hermitian matrices, the intersection between the cone of semidefinite matrices and the halfspace given by $\Tr[\vb{A}\vb{J}] \leq P_A$ is bounded and the $\log\det$ objective cannot diverge. 

In problem~\eqref{eq:direct_convex_bound}, we can use a simple sum of the drive power constraint \eqref{eq:tr_power_constraint} and resistive power conservation over the enter design domain [the imaginary part of \eqref{eq:tr_re_sca_constraint}] to construct such a linear trace inequality, with $\vb{A}$ given by
\begin{equation*}
    \mqty[\dmat{2\omega \alpha \Id_S, \Im(1/\chi^*) \Id_D}] + \mqty[\Asym(\G_{0,SS}) & \Im(\G_{0,SD}) \\ \Im(\G_{0,DS}) & \Asym(\G_{0,DD})].
\end{equation*}
The first matrix is a diagonal operator with positive entries and is positive definite. The second operator is positive semidefinite since it is a sub-operator of the positive semidefinite operator
\begin{equation}
    \mqty[\Asym(\G_{0,AA}) & \Asym(\G_{0,AA}) \\ \Asym(\G_{0,AA}) & \Asym(\G_{0,AA})]
    \label{eq:double_block_AsymG}
\end{equation}
where the domain $A = S \bigcup D$. To see that \eqref{eq:double_block_AsymG} is positive semidefinite, we start from the fact that $\Asym(\G_{0,AA})$ is positive semidefinite due to passivity~\cite{tsang_scattering_2004}, and therefore it has a Hermitian positive semidefinite matrix square root $\sqrt{\Asym(\G_{0,AA})}$. It suffices now to note that 
\begin{align*}
    &\mqty[\vb{x} \\ \vb{y}]^\dagger \mqty[\Asym(\G_{0,AA}) & \Asym(\G_{0,AA}) \\ \Asym(\G_{0,AA}) & \Asym(\G_{0,AA})] \mqty[\vb{x} \\ \vb{y}] \\
    &= \norm{\sqrt{\Asym(\G_{0,AA})} (\vb{x} + \vb{y})}_2^2 \geq 0. \numthis
\end{align*}

It is also worth pointing out that if the constraint matrix $\mA$ in \eqref{eq:PD_waterfilling} is indefinite, then the optimum diverges. To see this, consider any vector $\vb{h}$ outside the kernel of $\mH$, i.e., $\mH \vb{h} \neq \vb{0}$. If $\vb{h}^\dagger \mA \vb{h} < 0$, then we are done: $\mJ = \beta \vb{h}\vb{h}^\dagger$ will lead to a divergence of the objective as $\beta \rightarrow \infty$ while $\Tr[\mA \mJ] = \vb{h}^\dagger \mA \vb{h} \rightarrow -\infty < P_A$. If $\vb{h}^\dagger \mA \vb{h} \geq 0$, simply pick a vector $\vb{v}$ such that $\vb{v}^\dagger \mA \vb{v} < 0$; this is always possible given indefinite $\mA$. Then by constructing $\mJ = \beta_1(\vb{h}\vb{h}^\dagger + \beta_2 \vb{v}\vb{v}^\dagger)$ and choosing $\beta_2>0$ such that $\vb{h}^\dagger \mA \vb{h} + \beta_2 \vb{v}^\dagger \mA \vb{v} < 0$ the same argument as before holds: as $\beta_1 \rightarrow \infty$ the objective diverges thanks to the $\vb{h}\vb{h}^\dagger$ contribution while $\mJ \succ 0$ remains within the constraint halfspace. 

\section{Convex Relaxation Solution via Waterfilling Duality} \label{asec:SDP_numerics}

This appendix describes the details of solving the direct convex relaxation~\eqref{eq:direct_convex_bound} using an approach we dub the ``waterfilling dual'', which is a generalization of the solution technique used in \cite{ehrenborg_bounds_MIMO_2020}. Specifically, we seek the solution of a convex optimization problem over the $n \times n$ PSD matrix $\mJ$:
\begin{subequations}
    \begin{align}
        C = \max_{\mJ} \quad &f(\mJ) \equiv \log\det(\mI + \mH \mJ \mH^\dagger) \\
        \text{s.t.} \quad & \Tr[\mA \mJ] \leq P, \\
        &\Tr[\mB_l \mJ] = 0 \quad l \in [1,\cdots, m], \\
        & \mJ \succeq \vb{0},
    \end{align}
    \label{eq:general_waterfilling}
\end{subequations}
where $\mA \succ 0$ and we assume that the solution $\bar{\mJ}$ is non-zero, i.e., the $\mB_l$ and PSD constraints do not restrict $\mJ=\vb{0}$ and the PD $\mA$ constraint is needed for $C$ to be finite. While \eqref{eq:direct_convex_bound} is not explicitly of this form, Appendix \ref{asec:SDP_PD} shows how one can straightforwardly transform \eqref{eq:direct_convex_bound} into this form by taking linear combinations of constraints to obtain a PD constraint. The waterfilling dual approach to \eqref{eq:general_waterfilling} starts from the fact that a simplified version of \eqref{eq:general_waterfilling} without any $\Tr[\mB_l \mJ]=0$ constraints admits an analytical solution, i.e., waterfilling. Unfortunately, waterfilling does not directly generalize to multiple linear constraints, so we define the waterfilling dual function
\begin{subequations}
    \begin{align}
        \mathcal{D}(\vb{b}) = \max_{\mJ} \quad &\log\det(\mI + \mH \mJ \mH^\dagger) \\
        \text{s.t.} \quad & \Tr[(\mA + \sum_{l=1}^m b_l \mB_l) \mJ] \leq P, \\
        & \mJ \succeq \vb{0}.
    \end{align}
    \label{eq:waterfilling_dual}
\end{subequations}
The trace constraint in \eqref{eq:waterfilling_dual} is a linear superposition of the trace constraint in \eqref{eq:general_waterfilling}, so $\mathcal{D}(\vb{b}) \geq C$ for any combination of multipliers $\vb b$; furthermore, the value of $\mathcal{D}(\vb{b})$ can be evaluated using waterfilling. $\mathcal{D}$ can be understood as a \textit{partial} dual function of \eqref{eq:general_waterfilling} where we have not dualized the $\mJ \succeq \vb{0}$ constraint; $\mathcal{D}(\vb{b}) \geq C$ is thus a statement of weak duality. We now show explicitly that, as expected for a convex problem, \textit{strong duality} holds: the solution to the waterfilling dual problem
\begin{equation}
    \min_{\vb b} \quad \mathcal{D}(\vb b)
    \label{eq:min_waterfilling_dual}
\end{equation}
is exactly equal to $C$. To see this, consider the objective gradient at the primal optimal point $\grad f(\bar{\mJ})$. First order local optimality conditions imply that
\begin{equation}
    \grad f(\bar{\mJ}) = \left(\sum_{l=1}^m \bar{b}_l \mB_l\right) + a \mA + \sum s_k \vb{S}_k
    \label{eq:local_opt}
\end{equation}
where $a, s_k \geq 0$, and $\vb{S}_k$ are outward normal matrices of the PSD cone $\mathcal{S}^n_+$ locally at $\bar{\mJ}$: such a local description of $\mathcal{S}^n_+$ is possible due to the fact that any non-zero matrix on the boundary of $\mathcal{S}^n$ is an interior point of a face of $\mathcal{S}^n_+$~\cite{barvinok_convexity_2002}. Since by assumption the primal problem has a non-zero solution, the PD constraint must be active and therefore $a > 0$. Now, consider the dual value $\mathcal{D}(\bar{\vb b}/a)$ as the optimum of \eqref{eq:waterfilling_dual}: it is clear that $\bar{\mJ}$ is a feasible point of \eqref{eq:waterfilling_dual}, and furthermore by construction $\grad f(\bar{\mJ})$ satisfies the first order local optimality condition
\begin{equation}
    \grad f(\bar{\mJ}) = a \left(\mA + \sum_{l=1}^m \frac{\bar{b}_l}{a} \mB_l \right) + \sum s_k \vb{S}_k.
\end{equation}
Conclude that $\mathcal{D}(\bar{\vb b}/a) = f(\mJ) = C$ and therefore strong duality holds. Note also that for $\mathcal{D}(\vb b)$ to remain finite, $\mA + \sum b_l \mB_l$ must be positive semidefinite per Appendix \ref{asec:SDP_PD}, so we can always do waterfilling to evaluate finite $\mathcal{D}$ values. Finally, given the analytic form of the waterfilling solution of \eqref{eq:waterfilling_dual} one can also derive dual gradients $\partial \mathcal{D} / \partial \vb b$ using matrix perturbation theory, and thus the dual problem \eqref{eq:min_waterfilling_dual} can be solved using gradient-based optimization.

\section{Generalized Water-filling and Shannon Relaxation Solution} \label{asec:shannon_relax_proof1}
In this section, we provide a general solution of \eqref{eq:shannon_relax_1}.
The problem can be written
\begin{equation} \begin{aligned} \label{eq:simple_problem}
    \underset{\vb x, \vb y}{\max} \qquad &  f(\vb x, \vb y) = \sum_{i=1}^{n} \log_2 \left( 1 + \gamma x_i y_i \right) \\
    \text{s.t.} \qquad & \vb 1^T \vb x = \vb 1^T \vb y = 1, \\ 
    & \vb x, \vb y \geq 0,
\end{aligned} \end{equation}
where the \(c\) and prime in \(n_c, \gamma\) have been dropped, respectively. This can be thought as a ``generalized" water-filling problem where both \(\sigma_i^2\) and \(J_i\) are optimization degrees of freedom.
Let \(x^*\), \(y^*\) be optimal solutions to this problem, which exist because the function is continuous over the compact set defined by the constraints. 
First, we consider asymptotic solutions in \(\gamma\).
\\~\\
\(\gamma \ll 1\):
\(\sum_i \log_2 (1+\gamma x_i y_i) \approx \gamma/\log 2 \sum_i x_i y_i \). Given the constraints \(\vb 1^T x = 1\) and \(\vb 1^T y = 1\) (easy to see the optimum must have equality), and that \(\log(1+x) \leq x \), we get the bound
\begin{equation}
    C \leq \gamma / \log 2
\end{equation}
using the method of Lagrange multipliers, first on the maximization with respect to \(x_i\) and then \(y_i\). This is a bound for all \(\gamma\) but is most accurate when \(\gamma\) is small. 
\\~\\
\(\gamma \gg 1\):
In this case, we find that 
\(\sum_i \log_2 (1+\gamma x_i y_i) \approx \sum_i \log_2 x_i + \log_2 y_i + \log_2 \gamma \).
This maximization problem can also be analytically solved given the same constraints by the method of Lagrange multipliers to find
\begin{equation}
    C \to  n_c \log_2 \dfrac{1}{n_c} + n_c \log_2 \dfrac{1}{n_c} + n_c \log_2{\gamma} = n_c \log_2 \dfrac{\gamma}{n_c^2}.
\end{equation}
This corresponds to equal allocation \(1/n_c\) and \(1/n_c\) for all channels. 
This value approaches a bound on \(C\) as \(\gamma \to \infty\). 
Now, we show how this problem can be solved exactly, analytically.
In the following analysis, we take \(\log_2 \to \log\) for mathematical convenience, understanding that solutions are scaled by a constant factor. 

\begin{lemma}
    Let \(S_{x^*} = \{ i \in [n]: x_i^* > 0\}\) and \(S_{y^*} = \{i\in [n] : y_i^* > 0\}\) be the support sets of \(\vb x^*\) and \(\vb y^*\) respectively. Then \(S_{x^*} = S_{y^*} \equiv S\).
    \label{lemma:eq_support}
\end{lemma}
\begin{proof}
    Pick two indices \(i\neq j\). If \(x_i = 0, y_i = c\), then taking \(y_i \to 0\) and \(y_j \to y_j + c\) strictly increases the objective. 
    Intuition: allocating mass to some \(y_i\) if \(x_i=0\) contributes \(0\) to the objective, whereas allocating that same mass to another \(y_j\) contributes positively. Thus the support sets of \(\vb x^*\) and \(\vb y^*\) are equal.  
\end{proof}

\begin{lemma} \label{lem:conditions}
    For any pair of indices \(i, j \in S\),  \(x_i^* y_i^* x_j^* y_j^* = 1/\gamma^2 \) (Condition 1) or \(x_i^* = x_j^*\) (Condition 2).
\end{lemma}
\begin{proof}
Fix \(y^*\) and consider the following optimization problem 
\begin{equation} \begin{aligned}
    \max_{x} \qquad &f^{\gamma,y^*}(\vb x) = \sum_{i =1}^n \log \left( 1 + \gamma x_i y_i^* \right) \\
    \textrm{s.t.} \qquad & \vb 1^T \vb x = 1, \\ & \vb x \ge 0. 
\end{aligned} \label{eq:yaux_problem} \end{equation}
We know that $ x^*$ is a KKT point in the above optimization problem.
The KKT conditions are, for some \(\lambda > 0, \gamma_1, \dots \nu_n \ge 0 \),
\begin{equation} \begin{aligned}
    \lambda \vb 1 - \sum_{i=1}^n \nu_i \vb 1_{x_i^* = 0} &= \nabla f_{\gamma,y^*}(x^*),
\end{aligned} \end{equation}
where \(\vb 1_{x_i^* = 0}\) is an indicator function for \(x_i^* = 0\).
Now $x_i^* = 0$ if and only if $y_i^* = 0$ by Lemma~\ref{lemma:eq_support}. 
If $y_i^* = 0$ then it is easy to see that $\nabla f_i^{\gamma,y^*}(x^*) = 0$, and therefore $\nu_i = \lambda$. Otherwise, \(\vb 1_{x_i^* = 0}\) is \(0\) everywhere. 
\\~\\
Therefore, for each pair of \(i,j\in S\), 
\begin{equation} \label{eq:KKT2}
    \dfrac{y_i^* \gamma}{1+\gamma x^*_i y_i^*} = \dfrac{\gamma y_j^* }{1+\gamma x^*_j y_j^*}.
\end{equation}
Solving for \(x^*_j - x^*_i\), we obtain
\begin{equation}
    x^*_j - x^*_i = \dfrac{(y_j^* - y_i^*)}{\gamma y_j^* y_i^*}.
    \label{eq:yaux_xdiff}
\end{equation}
Making the same argument fixing \(x^*\), we find 
\begin{equation} \begin{aligned} \label{eq:kkt1}
    y_j^* - y_i^* &= \dfrac{(x_j^* - x_i^*)}{\gamma x_j^* x_i^*}. \\ 
\end{aligned} \end{equation}
Combining these, we find that 
\begin{equation}
    x_j^* - x_i^* = \dfrac{1/\gamma^2}{x_i^* y_i^* x_j^* y_j^*} (x_j^* - x_i^*),
\end{equation}
from which it follows that either \(x_i^* y_i^* x_j^* y_j^* = 1/\gamma^2\) or \(x_i^* = x_j^*\).
\end{proof}

\begin{lemma}
    For every $i$, $x_i^* = y_i^*$.
\end{lemma}

\begin{proof}

Enumerate the unique nonzero values of $x_i^*$ for $i=1\dots n$ by $X_1, \dots X_K$. Suppose $K \ge 2$ (if $K=1$ the result is clear). For $i=1 \dots k$ let $S_k = \{i \in \{1 \dots n\} \text{ such that } x_i=X_k\}$ for $k \in \{1 \dots K\}$.
By \eqref{eq:kkt1} if $x_i^* = x_j^*$ then also $y_i^* = y_j^*$. Denote the unique value of $y_i^*$ in the set $S_k$ by $Y_k$.
\\~\\
Now choose $i \in S_k$  for some $k \ge 2$ and $j \in S_1$. From \eqref{eq:KKT2} combined with the fact that $x_i^* y_i^* x_j^* y_j^* = 1/\gamma^2$, we observe
\begin{align*}
    \gamma y_i^* &= \gamma y_j^* \Big( \frac{1 + 1/(\gamma x_j^* y_j^*)}{1 + \gamma x_j^* y_j^*} \Big)
\end{align*}
and symmetrically in $x$,
\begin{align*}
    \gamma x_i^* &= \gamma x_j^* \Big( \frac{1 + 1/(\gamma x_j^* y_j^*)}{1 + \gamma x_j^* y_j^*} \Big).
\end{align*}
Now divide the above two equations to see that for every $k$,
\begin{align*}
    \frac{X_k}{Y_k} &= \frac{X_1}{Y_1}.
\end{align*}
Since $1=\sum_{i=1}^n x_i^* = \sum_{i=1}^n y_i^*$, this implies that $x_i^* = y_i^*$ for each $i$ such that $x_i^* > 0$.
\end{proof}
Now we know that for each $i\neq j$, $X_i^2 X_j^2 = 1/\gamma^2$. This is impossible if $K \ge 3$ (\(xy=c,yz=c,xz=c \implies y=z\)) so we know $K=1$ or $K=2$. 
First, we consider \(K=1\).
\begin{lemma} \label{lemma:derivatives}
    Let $g(x) = x \log (1 + 1/x^2)$. There is a unique local maximizer of $g$ over the real numbers, which is also the global maximizer, and is between $0$ and $1$. Let this real number be $r$.
\end{lemma}

\begin{proof}
    $g'(k\to 0) \to \infty$ and $g'(1) = \log(2) - 2/3 \le 0$. For $x \ge 1,$ $g'(x) = \log(1 + 1/x^2) - 2/(x^2 + 1) \le 1/x^2 - 2/(x^2 + 1) = (1 - x^2) / (x^2(x^2 + 1)) \le 0$. Also,
    \begin{align*}
        g''(x) &= \frac{2x}{x^2 + 1} - \frac{2x}{x^2} + \frac{4x}{(x^2 + 1)^2} = \frac{2x(x^2 - 1)}{x^2(x^2 + 1)^2}
    \end{align*}
    so $g''(x) \le 0$ for $x \in [0,1]$ and therefore $g$ is concave. So there can be only one local maximizer in this interval.
\end{proof}

\begin{lemma} \label{lemma:h}
    Let $h(k) = k \log (1 + \gamma/(k^2))$. The unique maximizer of $h$ over the integers is in the interval $[\lfloor r \sqrt{\gamma} \rfloor, \lceil r \sqrt{\gamma} \rceil ]$. 
\end{lemma}

\begin{proof}
    $h(k/\sqrt{N}) = \sqrt{\gamma} g(k)$ which is maximized over \(\mathbb{R}\) at \(r\sqrt{\gamma}\). Since the function is concave on the relevant interval, the solution is one of \( \lfloor r \sqrt{\gamma} \rfloor, \lceil r \sqrt{\gamma} \rceil \).
\end{proof}

Since \(h(k)\) is concave on \([0,n]\), if \(\floor{r\sqrt{\gamma}} \geq n\) then the optimal \(k\) is \(n\). Therefore we show that the \(K=1\) solution to~\eqref{eq:simple_problem} is
\begin{gather} 
n_{\mathrm{opt}}\log \left(1 + \dfrac{\gamma}{n_{\mathrm{opt}}^2} \right) \\ \quad n_{\mathrm{opt}} = 
\begin{cases}
     \floor{r\sqrt{\gamma}} \textrm{ or } \ceil{r\sqrt{\gamma}} \quad & \floor{r\sqrt{\gamma}} \leq n \\ 
     n \quad & \floor{r\sqrt{\gamma}} \geq n
\end{cases}
\end{gather}
where \(n_{\mathrm{opt}} \geq 1\), \(r = \max_{x\in [0, 1]} \left[ x \log \left(1+\dfrac{1}{x^2}\right) \right] \approx 0.505 \) and the floor and ceiling functions are a consequence of the number of populated channels being an integer allocation.
\\~\\
Now we consider $K=2$ solutions.
Suppose that $|S_1| = n_1$ and $|S_2| = n_2$ such that $n_1 + n_2 \le n$. 
Observe that $X_1 = \lambda/n_1$ and $X_2 = (1-\lambda)/n_2$ for some $\lambda \in [0,1]$. 
Consider the function,
\begin{equation} 
    f(\lambda) = n_1 \log \Big( 1 + \frac{\lambda^2 \gamma}{n_1^2} \Big) + n_2 \log \Big( 1 + \gamma \frac{(1 - \lambda)^2}{n_2^2} \Big). \label{eq:f}
\end{equation}
Any maximizer in the $K=2$ case must be of this form. 
We can explicitly find the maximizing $\lambda$ for this function in terms of $n_1$ and $n_2$.

\begin{lemma}
    The values of $\lambda$ which can be a maximum of $f$ in $[0,1]$ are $\lambda= 0$, $\lambda = 1$,
    \begin{align*}
        \lambda_0 = \frac{n_1}{n_2 + n_1} \quad & \textrm{ or } \quad \lambda_{\pm} = \frac{1 \pm \sqrt{1 - 4 n_1 n_2 / \gamma}}{2}.
    \end{align*}
\end{lemma}

\begin{proof}
    These are the values of $\lambda$ which have zero derivative, and those at the extremes of $\lambda \in [0, 1]$.
\end{proof}

These are the possible candidates for being a maximum of our function when $K=2$. The first and second cases correspond to $K=1$. In the third case, $X_1 = 1/(n_1 + n_2) = X_2$, so in fact $K = 1$. 
In the fourth case, we can assume without loss of generality that $n_1 \ge n_2$, meaning \(\lambda_+\) is the optimal \(K=2\) point (as it must be between local minima of \(\lambda_0\) and \(\lambda_-\)).
In practice, we numerically compare the \(K=1\) with all \(K=2\) solutions to get the global optima of this problem. 
In theory, we believe it is possible to prove that for all \(K=2\) solutions, there exists a \(K=1\) solution with a higher objective value. 

We now prove this for most values of \(N\). 
Towards a contradiction, we assume that $\lambda_+ \equiv \lambda, n_1, n_2$ in~\eqref{eq:f} represent a global maximum of the function~\eqref{eq:simple_problem}.
We will now show that $\lambda_+$ representing a global maximum leads to a contradiction. That is, we can either change $\lambda, n_1,$ or $n_2$ to a different value to get a better objective value.
\\~\\
\textbf{Case 1}: $1 \ge \max \left\{ (1/\sqrt{\gamma}) n_1 /\lambda, (1/\sqrt{\gamma}) n_2 / (1 - \lambda) \right\} $ \\
By the concavity of \(h\) (on the interval $[0, 1]$), we find that for $\lambda$ such that $1 \ge \dfrac{n_1 / \sqrt{\gamma}}{\lambda}$ and $1 \ge \dfrac{n_2 / \sqrt{\gamma}}{1 - \lambda}$,

\begin{align*}
    f(\lambda) &= \lambda h\Big( \frac{n_1}{\lambda} \Big) + (1 - \lambda) h\Big( \frac{n_2}{1 - \lambda} \Big)\\ &\le h(n_1 + n_2)
\end{align*}
 meaning we can get a better optimal value by summing $n_1$ and $n_2$ and considering a point in the $K=1$ case, giving a contradiction.
\\~\\
\textbf{Case 2:} $1 \ge 1/\sqrt{\gamma} n_1 /\lambda$, but $1 \le (1/\sqrt{\gamma}) n_2 / (1 - \lambda)$, and $n_2 \ge 2$. First, we prove the following lemma:
\begin{lemma} \label{lemma:decreasing}
    For $x \ge 1, g(x/2) \ge g(3x / 4) \ge g(x)$ and $g$ is strictly decreasing on $[3x / 4, x]$.
\end{lemma}

\begin{proof}
    We already showed that for $x \ge r$ $g$ is decreasing, and $r \le 3/4$ so the second inequality is trivial. 

    For the first inequality, for $x \ge 2 r$ again it is clearly trivial. We just need to show that it is true for $x \in [1, 2 r]$.
    Observe that $g(1/2) = 1/2 \log(5) \ge 0.7$ and for $x \le r$, $g$ is decreasing. Therefore for any $x \in [1/2, r]$, $g(x) \ge 0.7$. Also for any $x \ge 3/4$, again using the decreasing property of $g$, $g(x) \le g(3/4) \le 0.6$.
\end{proof}
Now, if $\dfrac{n_2/\sqrt{\gamma}}{1 - \lambda} \ge 1$ then $\dfrac{(n_2 - 1)/\sqrt{\gamma}}{1 - \lambda} \ge 1/2$ so, using Lemma \ref{lemma:decreasing}, $\lambda h(n_1 / \lambda) + (1 - \lambda) h (n_2 / (1 - \lambda)) \le \lambda h(n_1 / \lambda) + (1 - \lambda) h ((n_2 - 1) / (1 - \lambda))$.
Therefore, \(n_2 \to n_2 - 1\) increases the objective value, violating the assumption of optimality. 
\\~\\
\textbf{Case 3:} $ 1 \ge  n_1 / (\sqrt{\gamma}\lambda)$, but $1 \le (n_2 / \sqrt{\gamma}) / (1 - \lambda)$, and $n_2 = 1$. 
\\~\\
Note that in this case we must have $n_1 = \lfloor r \sqrt{\gamma} \rfloor / \sqrt{\gamma} $ because otherwise $n_1 \gets n_1 + 1$ would yield an objective improvement. 
\\~\\
First we must observe that $n_2 / (\sqrt{\gamma} \lambda) \ge 2$. This can be done by numerically solving the following bounded three-dimensional program and observing that the optimal objective value is less than $0$,
\begin{align*}
    \underset{x \le 1, y \in [1, 2], \lambda \ge 0.5}{\text{maximize}} \lambda g(x) + (1 - \lambda) g(y) - g(\lambda x + (1 - \lambda) y ). 
\end{align*}
Now using this fact, plus the fact that $g$ is concave on $[0, 1]$,
\begin{align*}
    & \quad \lambda g \Big( \frac{n_1/\sqrt{\gamma}}{\lambda} \Big) + (1 - \lambda) g \Big( \frac{1/\sqrt{\gamma}}{1 - \lambda} \Big)  \\
    &\le \lambda g \Big( \frac{n_1}{\lambda \sqrt{\gamma}} \Big) + (1 - \lambda) g (2) \\
    &= \lambda  \Bigg[ g(n_1 /\sqrt{\gamma}) + g \Big( \frac{n_1}{\lambda \sqrt{\gamma}} \Big) - g(\dfrac{n_1}{\sqrt{\gamma}}) \Bigg] + (1 - \lambda) g(2) \\
    & \le \lambda  \Bigg[ g(n_1 /\sqrt{\gamma}) + g'(n_1 /\sqrt{\gamma}) \Big( \frac{n_1}{\lambda \sqrt{\gamma}} - \dfrac{n_1}{\sqrt{\gamma}}\Big) \Bigg] + (1 - \lambda) g(2) \\
    & = g(\dfrac{n_1}{\sqrt{\gamma}}) \lambda + \dfrac{(1 - \lambda) n_1}{\sqrt{\gamma}} g'(n_1 /\sqrt{\gamma}) + (1 - \lambda) g(2).
\end{align*}
So an optimal point of the $K=2$ type is only possible if
\begin{align*}
    & \underset{m}{\max} g(m/\sqrt{\gamma}) \\ & \le g(n_1 /\sqrt{\gamma}) \gamma + (1 - \gamma) (n_1 /\sqrt{\gamma}) g'(n_1 /\sqrt{\gamma}) + (1 - \gamma) g(2)
\end{align*}
or,
\begin{align*}
    g'(n_1 /\sqrt{\gamma}) & \ge \Big(\frac{\underset{m}{\max}g(m /\sqrt{\gamma}) - \lambda g(n_1 /\sqrt{\gamma})}{1 - \lambda}- g(2)\Big) \frac{\sqrt{\gamma}}{n_1} \\
    & \ge (g(n_1 /\sqrt{\gamma})- g(2)) \frac{\sqrt{\gamma}}{n_1}.
\end{align*}
We can now look at $x$ for which,
\begin{align*}
    g'(x) \ge (g(x) - g(2)) \frac{1}{x}.
\end{align*}
These are $x \le 0.235$ and $x \ge 4.245$ (solved via a numerical solver). This means that to get an optimal point of this form we need $\lfloor r / \sqrt{N} \rfloor \sqrt{N} = n_1 \sqrt{N} \le 0.235$ which, since $r > 0.5$, is impossible.\\~\\

In conclusion, if both $\dfrac{n_1}{\sqrt{\gamma}\lambda}$ and $\dfrac{n_2 / \sqrt{\gamma}}{1-\lambda} \le 1$, by case 1 there exists a value with \(K=1\) with a better objective value. If either $\dfrac{n_1}{\lambda \sqrt{\gamma}}$ or $\dfrac{ n_2 / \sqrt{\gamma}}{1 - \lambda}$ are greater than 1 (WLOG the second case), with $n_2 \ge 2$, by case 2 this allocation is not optimal. 
Finally, if both $\frac{ n_1}{\sqrt{\gamma} \lambda}$ and $\frac{ n_2 / \sqrt{\gamma}}{1 - \lambda}$ are greater than \(1\), then numerically optimal \(K=2\) points are not possible. 
\(K=2\) points are nevertheless numerically checked to not be optimal in the examples presented in the main text. 

\section{Energy-Conserving QCQP Constraints} \label{asec:QCQP}

It is possible to derive the energy conservation constraints presented in the main text via operator constraints (which, like Maxwell's equations are dependent on the spatial profile of permittivity) relaxed to structure-agnostic scalar constraints~\cite{chao_physical_2022}. 
Although these constraints encode less information than Maxwell's equations in full generality, they allow the relaxed problem to be formulated as a QCQP from which techniques in the main text (i.e., Lagrange duality) can be employed to calculate structure-agnostic bounds on photonic performance. 
In this section, we show how these constraints can be directly derived from energy conservation as expressed by Poynting's theorem for time-harmonic complex fields~\cite{tsang_scattering_2004}:
\begin{multline}
    \int_{\partial V} \dd\vb{\sigma} \cdot (\vb{E}\cross\vb{H}^*) = i\omega \int_V (\vb{H}^*\cdot\mu\cdot\vb{H} -  \vb{E}\cdot\epsilon^*\cdot\vb{E}^*) \,\dd V \\
    - \int_V \vb{E}\cdot\vb{J}^* \,\dd V. 
    \label{eq:Poynting_cplx}
\end{multline} 
For simplicity, we will assume a non-magnetic material $\mu=\mu_0=1$ and scalar isotropic permittivity and susceptibility $\epsilon = 1 + \chi$, though the derivation is valid for anisotropic $\epsilon$ as well~\cite{chao_maximum_2022}. 

Consider a scattering theory picture where a free current source $\vb{J}_v$ generates the fields $\vb{E}_v$, $\vb{H}_v$ in vacuum and $\vb{E}_t$, $\vb{H}_t$ in the presence of a structure with material distribution $\epsilon(\vb{r}) = 1 + \chi \I_s(\vb{r})$ where $\I_s(\vb{r})$ is an indicator function. There is an induced polarization current $\vb{J}_s$ in the material which produces scattered fields $\vb{E}_s$ and $\vb{H}_s$ that combine with the vacuum fields to give the total field: $\vb{E}_t = \vb{E}_v+\vb{E}_s$, $\vb{H}_t = \vb{H}_v + \vb{H}_s$. The complex Poynting theorem thus \eqref{eq:Poynting_cplx} applies to three sets of currents, fields, and environments: $(\vb{J}_v, \vb{E}_v, \vb{H}_v)$ in vacuum, $(\vb{J}_s, \vb{E}_s, \vb{H}_s)$ in vacuum, and $(\vb{J}_v, \vb{E}_t, \vb{H}_t)$ over the structure, giving
\begin{multline}
    \int_{\partial V_k} \,\dd \vb{\sigma}\cdot(\vb{E}_v\cross\vb{H}^*_v) = i\omega \int_{V_k} \vb{H}_v^* \cdot \vb{H}_v \,\dd V \\- i\omega  \int_{V_k} \vb{E}_v \cdot \vb{E}_v^* - \int_{V_k} \vb{E}_v \cdot \vb{J}_v^* \,\dd V .
    \label{eq:Poynting_inc}
\end{multline}

\begin{multline}
    \int_{\partial V_k} \,\dd \vb{\sigma}\cdot(\vb{E}_s\cross\vb{H}^*_s) = i\omega \int_{V_k} \vb{H}_s^* \cdot \vb{H}_s \,\dd V \\- i\omega \int_{V_k} \vb{E}_s \cdot \vb{E}_s^* \,\dd V - \int_{V_k} \vb{E}_s \cdot \vb{J}_s^* \,\dd V .
    \label{eq:Poynting_sca}
\end{multline}

\begin{multline}
    \int_{\partial V_k} \,\dd \vb{\sigma}\cdot(\vb{E}_t \cross \vb{H}^*_t) 
    = i\omega \int_{V_k} \vb{H}_t^* \cdot \vb{H}_t \,\dd V \\- i\omega \int_{V_k} (1 + \chi^*\I_s) \vb{E}_t \cdot \vb{E}_t^* \,\dd V - \int_{V_k} \vb{E}_t \cdot \vb{J}_v^* \,\dd V .
    \label{eq:Poynting_tot}
\end{multline}
Subtracting \eqref{eq:Poynting_inc} and \eqref{eq:Poynting_sca} from \eqref{eq:Poynting_tot} gives
\begin{multline}
    \Big\{ i \omega \int_{V_k} \vb{H}_v^* \cdot \vb{H}_s \,\dd V - \int_{\partial V_k} \,\dd \vb{\sigma}\cdot(\vb{E}_s \cross \vb{H}_v^*) \\ - i\omega \int_{V_k} \vb{E}_s \cdot \vb{E}_v^* \,\dd V  \Big\} \\
    +\Big\{i \omega \int_{V_k} \vb{H}_s^* \cdot \vb{H}_v \,\dd V - \int_{\partial V_k} \,\dd \vb{\sigma}\cdot(\vb{E}_v \cross \vb{H}_s^*) \\ - i\omega \int_{V_k} \vb{E}_v \cdot \vb{E}_s^* \,\dd V  \Big\} \\
    = \int_{V_k} \vb{E}_s \cdot \vb{J}_v^* \,\dd V - \int_{V_k} \vb{E}_s \cdot \vb{J}_s^* \,\dd V + i\omega \int_V \chi^* \vb{E}_t \cdot \vb{E}_t^* \,\dd V \numthis.
    \label{eq:Poynting_intermediate}
\end{multline}
Now, using vector calculus identities along with the Maxwell wave equations $\curl\curl\vb{E}_v - \omega^2 \vb{E}_v = i\omega \vb{J}_v$ and $\curl\curl\vb{E}_s - \omega^2 \vb{E}_s = i\omega \vb{J}_s$, the two curly brackets in \eqref{eq:Poynting_intermediate} can be shown to be equal to $\int_{V_k} \vb{E}_s \cdot \vb{J}_v^* \,\dd V$ and $\int_{V_k} \vb{E}_v \cdot \vb{J}_s^* \,\dd V$ respectively. Finally, the induced current $\vb{J}_s$ can be swapped out by the polarization $\vb{p}$ via $\vb{J}_s = -i\omega \vb{p}$, and the scattered field $\vb{E}_s = \bmm{G} \vb{p}$, to give
\begin{equation}
    \int_{V_k} \vb{E}_v^* \cdot \vb{p} \,\dd V = \int_V \chi^{-1*} \vb{p}^*\cdot \vb{p} \,\dd V - \int_{V_k} \vb{p}^* \cdot ( \bmm{G}^\dagger \vb{p}) \,\dd V. 
\end{equation}
This can be written in a compact operator notation
\begin{equation}
    \vb{E}_v^\dagger \I_{V_k} \vb{p} = \vb{p}^\dagger (\chi^{-\dagger} - \bmm{G}^\dagger) \I_{v_k} \vb{p},
    \label{eq:QCQP_constraint}
\end{equation}
giving a form of the energy conservation constraints in the main text, where \(\vb E_v \to \vb S\). 
From this derivation it is clear that the constraint \eqref{eq:QCQP_constraint} encodes conservation of power during the electromagnetic scattering process for every region $V_k$. 
In the case of many sources (as in this paper), each source defines a different scattering problem, with individual source-polarization pairs $(\vb{S}_j, \vb{p}_j)$ satisfying constraints of the form \eqref{eq:QCQP_constraint}. There are also additional ``cross-constraints'' that capture the fact that the same structured media generates the $\vb{p}_j$ induced in each case:
\begin{equation}
    \vb S_j^\dagger \vb p_k - \vb p_j^\dagger \left(\chi^{-\dagger} - \bmm{G}^\dagger \right) \vb p_k = 0, \quad \forall j,k.
    \label{eq:cross_constraints}
\end{equation}
For a more detailed discussion of these constraints, we refer the reader to~\cite{molesky_mathbbt-operator_2021}.
For computational simplicity, only \(j=k\) constraints are enforced in the bounds calculated in this paper, although constraints are enforced for all computational voxels \(V_k\); future incorporation of these constraints is expected to tighten bounds.

\section{Physical Bounds on Green's Function Channel Capacities} \label{asec:channel_bounds}

We will now show how the Lagrange dual relaxation as described in Refs.~\cite{chao_physical_2022, molesky_mathbbt-operator_2021} can be utilized to obtain limits on the Frobenius norm and largest singular value of the total Green's function under arbitrary structuring. 
Ultimately, we seek to place sum-rule bounds on the squares of the singular values of \(\G_{t,RS}\). 
These limits are combined with solutions to Eq.~\eqref{eq:shannon_relax_1} to find limits on the Shannon capacity. 

Unlike the prior algebraic relaxations exploiting passivity (e.g.,~\cite{venkataram_fundamental_2020-2}), these incorporate significantly richer physics through the enforcement of local energy conservation constraints and thus yield tighter limits. 
Note that the trace of \(\G_{t,RS}^\dagger \G_{t,RS}\) can be evaluated by \(\sum_j \vb{v}_j^\dagger \G_{t,RS}^\dagger\G_{t,RS} \vb{v}_j\) in the basis of the singular vectors of \( \G_{0,RS} = \sum_j s_j \vb{u}_j \vb{v}_j^\dagger\).
Defining the vacuum ``source" fields in the receiver region \(\vb S_j \equiv \G_0 \vb v_j = s_j \vb u_j\) with amplitude \(\abs{s_j}^2\) and the polarization field in the photonic structure \(\vb p_j \equiv \left( \chi(\vb r)^{-1} - \G_0 \right)^{-1} \vb S_j \equiv \T \vb S_j\)~\cite{molesky_mathbbt-operator_2021}, we find 
\begin{multline} \label{eq:obj} 
    \vb{v}_j^\dagger \G_{t,RS}^\dagger\G_{t,RS} \vb{v}_j = \vb{v}_j^\dagger \G_{0,RS}^\dagger\G_{0,RS} \vb{v}_j + 
     \vb{p}_j^\dagger \G_{0,RD}^\dagger \G_{0,RD} \vb{p}_j \\ + 
     2\Re{\vb{v}_j^\dagger \G_{0,RS}^\dagger \G_{0,RD} \vb{p}_j } 
\end{multline}
where \(\G_{0,RD}\) acts on polarization fields in the \textit{design region} D (which may contain any region) and gives fields in the receiver region. 
The first term describes the square of the \(j\)-th singular value of the vacuum Green's function \(\abs{s_j}^2 \), while the rest describe the contribution from the photonic structure.
In this picture, \(\vb p_j, \vb S_j\) are vectors in the design region, while \(\vb v_j\) are vectors in the source region.
Generalized energy conservation scalar identities~\cite{chao_physical_2022}, can be written 
\(\vb S_j^\dagger \I_v \vb p_k - \vb p_j^\dagger \left(\chi^{-\dagger} - \G_{0,DD}^\dagger \right) \I_v \vb p_k = 0\), \(\forall j,k\) (see Appenidx~\ref{asec:QCQP}), where \(\I_v\) is projection into any spatial subdomain. 
For the purposes of clarity, we take \(\I_v \to \I\), noting that constraints will be enforced in every region for the numerical calculation of bounds.
From these components, quadratically constrained quadratic programs (QCQPs) can be written to maximize different objectives related to the singular values of \(\G_{t,RS}\), from which bounds on their optimal values can be computed with duality relaxations as described below. 
\\~\\
\textbf{Frobenius norm \(\lVert\mathbb{G}_{t,RS}\rVert_{\mathtt{F}}^{2} = \Tr(\G_{t,RS}^\dagger\G_{t,RS}) = \sum_i \abs{\sigma_i}^2 \leq M_\mathtt{F}\):}
\begin{equation} \label{eq:frobenius_bound} \begin{aligned} 
    \max_{\vb p_j} \quad & \sum_j  \vb{v}_j^\dagger \G_{t,RS}^\dagger\G_{t,RS} \vb{v}_j  \\
     \textrm{s.t.} \quad & \vb S_j^\dagger \vb p_k - \vb p_j^\dagger \left(\chi^{-\dagger} - \G_{0,DD}^\dagger \right) \vb p_k = 0, \\
    & \forall j,k \\
\end{aligned} \end{equation}
\textbf{$L_2$ norm \(\lVert\mathbb{G}_{t,RS}\rVert_{2}^{2} = \sigma_\mathrm{max} \leq M_1\):}
\begin{equation}
    \begin{aligned}
        \max_{\vb{v},\vb{p}} \quad & \vb{v}^\dagger \G_{t,RS}^\dagger\G_{t,RS} \vb{v}  \\
        \textrm{s.t.} \quad & \vb{v}^\dagger \G_{0,DD}^\dagger \vb{p} - \vb{p}^\dagger (\chi^{-1\dagger} - \G_{0,DD}^\dagger) \vb{p} = 0, \quad \\ 
        &\vb{v}^\dagger \vb{v} = 1.
    \end{aligned}
\end{equation}

The first problem is simply maximizing the Frobenius norm of \(\G_{t,RS}\) subject to energy conservation constraints.  In the second problem, we maximize the largest singular value by co-optimizing the source currents \(\vb v\) and the polarization currents \(\vb p\) and enforcing that the singular vector is normalized. 
We note that in all calculations, ``cross"-constraints between sources \(j\neq k\) are not enforced for computational efficiency. 
This relaxes the problem: full incorporation of these constraints is expected to tighten limits. 

To compute limits on these problems, we note that the Lagrangian $\mathcal{L}$ of a QCQP with primal degrees of freedom \(\bm \psi\) is given by
\begin{equation}\begin{aligned}
    \Lag(\bm \psi, \lambda) =  -\bm \psi^\dagger \mathbb{A}(\lambda) \bm \psi &+ 2 \Re \left( \bm \psi^\dagger \mathbb{B}(\lambda) \vb S \right) \\ &+ \vb S^\dagger \mathbb{C}(\lambda) \vb S
\end{aligned}\end{equation}
where \(\mathbb{A}, \mathbb{B}\), and \(\mathbb{C}\) represent the quadratic, linear, and constant components of the Lagrangian and contain both objective and constraint terms.  
The Lagrange dual function $\mathcal{G}$ and its derivatives can be written, in terms of \(\bm \psi_{\mathrm{opt}} \equiv \mathbb{A}^{-1} \mathbb{B} \vb S\) and for positive-definite \(\mathbb{A}\),

\begin{align}
    \mathcal{G}(\lambda) &= \vb S^\dagger \left( \mathbb{B}^\dagger \mathbb{A}^{-1} \mathbb{B} + \mathbb{C} \right) \vb S, \\
    \dfrac{\partial\mathcal{G}}{\partial \lambda} &= 2 \Re \left( \bm \psi_{\mathrm{opt}}^{\dagger} \dfrac{\partial \mathbb{B}}{\partial \lambda} \vb S \right) - \bm \psi_{\mathrm{opt}}^{\dagger} \dfrac{\partial \mathbb{A}}{\partial \lambda} \bm \psi_{\mathrm{opt}} \\ &  \quad+ \vb S^\dagger \dfrac{\partial \mathbb{C}}{\partial \lambda} \vb S.\nonumber
    \label{eq:dual_ZTT}
\end{align}
The positive-definiteness of \(\mathbb{A}\) ensures that the dual problem has feasible points with finite objective values. 
We must therefore show that for each primal problem, there exist dual feasible Lagrange multipliers \(\lambda\) such that \(\mathbb{A}\) is positive-definite. 
In the case of \textbf{(1)}, we simply note that \(\Asym(\G_{0,DD})\) is positive-semidefinite~\cite{tsang_scattering_2004}. 
Therefore, for a lossy design material \(\Im \chi > 0\), we can increase the Lagrange multiplier for the imaginary part of this constraint to ensure \(\mathbb{A}\) is positive-definite at some \(\lambda\). All other multipliers can be initialized to zero. 
In the case of \textbf{(2)}, the same strategy can be employed to make \(\mathbb{A}\) positive-definite in the \(\vb p, \vb p\) sub-block. 
To ensure positive-definiteness in the \(\vb v, \vb v\) sub-block, the semidefinite quadratic constraints \(\vb v_i^\dagger \vb v_i = 1\) can be employed. 
Overall, the existence of these points proves the existence of a bound on their respective primal problems.

\bibliography{refs}

\end{document}